\documentclass[11pt]{article}
  \usepackage[utf8]{inputenc} 
  \usepackage[T1]{fontenc}
  \usepackage{fullpage}
  \usepackage{amsthm, latexsym,amssymb,amsfonts,amsmath,mathrsfs,float}
  \usepackage[font=small,labelfont=bf, width=.75\textwidth]{caption} 
  \usepackage[dvipsnames]{xcolor}
  \usepackage{mathtools}
  \usepackage{graphicx}
  \usepackage{dsfont}
  \usepackage{subcaption}
  \usepackage{makecell}
  \usepackage{qcircuit}
  \usepackage{tikz}
  \usepackage{mleftright}
  \mleftright
  \usepackage[toc,page]{appendix}
  \usepackage{bbm}
  \usetikzlibrary{decorations.pathmorphing}
  \usepackage[linktocpage=true, linkbordercolor=red, ocgcolorlinks]{hyperref}
  \hypersetup{colorlinks=true, allcolors=QuSoft}
  \usepackage{enumitem}
  \usepackage[affil-it]{authblk}
  \makeatletter
  \def\namedlabel#1#2{\begingroup
	#2%
	\def\@currentlabel{#2}%
    	\phantomsection\label{#1}\endgroup
  }   

  \definecolor{UltramarineBlue}{RGB}{18,10,143}
  \definecolor{QuSoft}{RGB}{179, 16, 32}

  \newcommand{\ignore}[1]{}	

  \usepackage{braket}
 \newcommand{\ketbra}[2]{\ket{#1}\!\bra{#2}}

  \DeclareMathOperator{\Tr}{Tr}

  \newcommand{\Prob}{\mathbb{P}}
  \newcommand{\eps}{\varepsilon}

  \newtheorem{theorem}{Theorem}[section]

  \newtheorem{proposition}[theorem]{Proposition}
  \newtheorem{lemma}[theorem]{Lemma}

  \makeatletter
\renewcommand*{\@fnsymbol}[1]{\ensuremath{\ifcase#1\or *\or \mathparagraph\or \ddagger\or
    \mathsection \or \dagger\or \|\or **\or \dagger\dagger
    \or \ddagger\ddagger \else\@ctrerr\fi}}
\makeatother
  
\begin{document}
\numberwithin{equation}{section}

\title{\textbf{Towards Practical and Error-Robust Quantum Position Verification}}

\author[1]{Rene Allerstorfer\footnote{Email: \href{mailto:rene.allerstorfer@cwi.nl}{rene.allerstorfer@cwi.nl}}}
\author[1,2]{Harry Buhrman\footnote{Email: \href{mailto:harry.buhrman@cwi.nl}{harry.buhrman@cwi.nl}}}

\author[2]{Florian Speelman\footnote{Email: \href{mailto:f.speelman@uva.nl}{f.speelman@uva.nl}}}
\author[1]{Philip Verduyn Lunel\footnote{Email: \href{mailto:philip.verduyn.lunel@cwi.nl}{philip.verduyn.lunel@cwi.nl}}}
\affil[1]{QuSoft, CWI Amsterdam, Science Park 123, 1098 XG Amsterdam, The Netherlands}
\affil[2]{QuSoft, University of Amsterdam, Science Park 904, 1098 XH Amsterdam, The Netherlands}

\date{Dated: \today}
\maketitle

\begin{abstract}
\noindent 
Loss of inputs can be detrimental to the security of quantum position verification (QPV) protocols, as it may allow attackers to not answer on all played rounds, but only on those they perform well on. In this work, we study \textit{loss-tolerant} QPV protocols. We propose a new fully loss-tolerant protocol QPV$_{\textsf{SWAP}}$, based on the SWAP test, with several desirable properties. The task of the protocol, which could be implemented using only a single beam splitter and two detectors, is to estimate the overlap between two input states. By formulating possible attacks as a semi-definite program (SDP), we prove full loss tolerance against unentangled attackers restricted to local operations and classical communication, and show that the attack probability decays exponentially under parallel repetition of rounds. We show that the protocol remains secure even if unentangled attackers are allowed to quantum communicate, making our protocol the first fully loss-tolerant protocol with this property. A detailed analysis under experimental conditions is conducted, showing that QPV$_{\textsf{SWAP}}$ remains fairly robust against equipment errors. We identify a necessary condition for security with errors and simulate one instance of our protocol with currently realistic experimental parameters, gathering that an attack success probability of $\leq10^{-6}$ can be achieved by collecting just a few hundred conclusive protocol rounds.
\end{abstract}

\section{Introduction}
Geographical position is an important contributor to trust---for example, a message which provably comes from a secure location in a government institution, has automatic credence to actually be sent by that government.
Position-based cryptography is the study of using position as a cryptographic credential.
The most basic task here is to certify someone's position, but this can be extended to messages that can only be read at a certain location, or to \emph{authenticating} that a message came (unaltered) from a certain location.

We will focus on the task of \emph{position verification}, which can be used as the building block for tasks like position-based authentication. For simplicity, the focus will be on the one-dimensional case, i.e.\ verifying one's position on a line, but the relevant ideas generalize readily to more dimensions. In our case, protocols will have the form of two \emph{verifiers}, $\mathsf{V_0}$ and $\mathsf{V_1}$, attempting to verify the location of a \emph{prover} $\mathsf{P}$.
An adversary to a scheme will take the form of a coalition of attackers, while the location of $\mathsf{P}$ is empty.
Notationally, we'll use $\mathsf{A}$ (or Alice) for the attacker located between $\mathsf{V_A}$ and the location of $\mathsf{P}$, and $\mathsf{B}$ (or Bob) for the attacker location between the location of $\mathsf{P}$ and $\mathsf{V_B}$. 

It was shown by Chandran, Goyal, Moriarty, and Ostrovsky~\cite{chandran_position_2009} that without any additional assumptions, position verification is an impossible task to achieve classically.
The quantum study of quantum position verification (QPV) was first initiated by Beausoleil, Kent, Munro, and Spiller resulting in a patent pusblished in 2006 \cite{KentPatent2006}. The topic first appeared in the academic literature in 2010 ~\cite{Malaney2010, MalaneyNoisy2010}, followed by various proposals and ad-hoc attacks ~\cite{kent_quantum_2011,lau_insecurity_2011}.
A general attack on quantum protocols for this task was presented by Buhrman, Chandran, Fehr, Gelles, Goyal, Ostrovsky, and Schaffner~\cite{buhrman_position-based_2011}, requiring a doubly-exponential amount of entanglement.
This attack was further improved to requiring an exponential amount of entanglement by Beigi and K\"onig~\cite{beigi_simplified_2011} -- much more efficient but still impractically large. (See also~\cite{gao2013,dolev2019} for generalizations of such attacks to different settings, with similar entanglement scaling.)

A natural question is therefore whether some QPV protocols can be proven secure against attackers that share a limited amount of entanglement or even none at all.

Since it's so hard to generate entanglement, it is already interesting to study whether protocols are secure against adversaries that are very limited in their access to pre-shared entangled states.
For instance, the QPV$_{\text{BB84}}$ protocol \cite{kent_quantum_2011}, inspired by the BB84 quantum key-distribution protocol, involves only a single qubit sent by  $\mathsf{V_A}$, in the state $\ket{0}$, $\ket{1}$, $\ket{+}$, or $\ket{-}$, and the choice of basis sent by  $\mathsf{V_B}$. 
Even though this protocol is insecure against attackers sharing a single EPR pair~\cite{lau_insecurity_2011}, security can be proven against unentangled attackers \cite{buhrman_position-based_2011}, so that $\Theta({n})$ entanglement is required to break the $n$-fold parallel repetition \cite{tomamichel_monogamy--entanglement_2013,ribeiro_tight_2015}.
At the current technological level, such protocols are very interesting to analyze, and would already give a super-classical level of security if implemented in practice.

Additionally, other protocols have been proposed \cite{kent_quantum_2011,chakraborty_practical_2015, unruh_quantum_2014,junge2021geometry,bluhm2021position}, that combine classical and quantum information in interesting ways, sometimes requiring intricate methods to attack~\cite{buhrman2013garden, speelman2016,olivo_breaking_2020}. 

Unfortunately, implementing any of the mentioned protocols would run into large obstacles: the quantum information involved would have to be sent at the speed of light, i.e., using photons, and in realistic experimental setups a large fraction of photons will be lost and errors occur.
Compensating for this in the most natural way, by ignoring rounds whenever the prover claims that a photon was lost in transmission, lets attackers break all these protocols because they are not fully loss tolerant.
In our contribution, we study loss-tolerant QPV by presenting a new fully loss-tolerant protocol, together with a comprehensive security analysis both in the theoretical and experimental setting.

\paragraph{Loss-tolerance in QPV.} Throughout, we will use $\eta$ as rate of transmission, i.e., the probability that an quantum message arrives -- in realistic protocols.
We will distinguish two types of loss tolerance that we might require schemes to satisfy.

The first, \emph{partial loss tolerance}, refers to a protocol which is secure for some values $\eta \geq \eta_{\mathrm{threshold}}$, meaning that the honest parties have a maximum level of allowed loss. Security is only guaranteed in a situation where a high enough fraction of the rounds are played.
If significantly more photons than this threshold are lost, then the protocol will have to abort. Examples of partial loss tolerant schemes are extensions of QPV$_{\text{BB84}}$ to more bases \cite{qi_loss-tolerant_2015,speelman_position-based_2016}, that are secure against unentangled attackers in an environment with some loss.\footnote{This notion will be satisfied to a small level even by schemes that are not designed to be loss tolerant, simply by having some error-robustness. The basic QPV$_{\text{BB84}}$ scheme can directly be seen to be partially loss tolerant for loss below $\frac{1}{2}-\frac{1}{2\sqrt{2}}$, and the simplest attack that uses loss only works when the loss is above $\frac{1}{2}$.}

\emph{Full loss tolerance} is achieved when a protocol is secure, irrespective of the loss rate. In particular, the protocol stays secure when conditioning on those rounds where the prover replied, fully ignoring rounds where a photon is lost. The protocol by Lim, Xu, Siopsis, Chitambar,
Evans, and Qi \cite{qi_free-space_2015,lim_loss-tolerant_2016}, the first fully loss-tolerant protocol, consists of  $\mathsf{V_A}$ and $\mathsf{V_B}$ both sending a qubit, and having the prover perform a Bell measurement on both, broadcasting the measurement outcome.
This protocol is secure against unentangled attackers, no matter the loss rate.

In the current work, we advance the study of loss-tolerant QPV with the following results:
\begin{itemize}
	\item We present a new fully-loss tolerant protocol: QPV$_{\textsf{SWAP}}$, which is based on the SWAP test~\cite{buhrman_quantum_2001}. The new protocol compares favorably to Lim et al.'s protocol~\cite{lim_loss-tolerant_2016} in terms of ease of implementation using linear optics, by requiring only a single, non-polarizing beam splitter -- the Hong-Ou-Mandel effect can be viewed as equivalent to the SWAP test \cite{jex2004comparing, garcia-escartin_swap_2013} so that, physically speaking, our protocol is based on two-photon interference.\footnote{The protocol uses two input photons, one generated by each verifier.}
	\item We prove fully loss tolerant security by formulating possible attacks as a semi-definite program (SDP), and show that the protocol is secure against unentangled attackers who can communicate only classically.
	
	Additionally, we show that the attack probability decays exponentially under \emph{parallel repetition}: when attackers respond to a size $k$ subset out of $n$ parallel rounds, pretending photon loss on the other inputs, their probability of a successful attack still decays exponentially in~$k$.
	Such a parallel repetition is not known for the protocol of~\cite{lim_loss-tolerant_2016}, and this is the first parallel repetition theorem for fully loss tolerant QPV. We obtain this result by constructing an SDP formulation of the $n$-fold parallel repetition of the problem, constructing a dual of this SDP for variable $n$, and then finding a point in the generalized dual problem.
	\item We show that the SWAP-test can be perfectly simulated with local operations and one round of classical communication if one maximally entangled state is pre-shared. Hence $\tilde{O}(n)$ EPR pairs are sufficient for an entanglement attack on our $n$-round protocol. We also show that at least $\sim 0.103n$ EPR pairs are necessary.
	\item Using an argument based on the monogamy of entanglement from \cite{AllBuhSpeVer22}, security of QPV$_{\textsf{SWAP}}$ in the setting where unentangled attackers can quantum communicate is shown, making our protocol the first fully loss-tolerant QPV protocol with this property.
	\item We provide a detailed analysis of our protocol under experimental conditions, treating all equipment errors that can occur in the setup -- from source to detection. We show that QPV$_{\textsf{SWAP}}$ remains fairly robust against equipment errors, making it a great candidate for practical QPV. A necessary condition for security with errors of our protocol is identified and we simulate a specific QPV$_{\textsf{SWAP}}$ with either identical or orthogonal inputs using currently realistic experimental parameters. We gather that an attack success probability of $\leq 10^{-6}$ can be achieved by collecting just a few hundred conclusive protocol rounds.
\end{itemize}

\subsection{Structure of the paper}
In Section~\ref{sec: qpvswap} we present the protocol QPV$_{\textsf{SWAP}}$, with primary security analysis in Section~\ref{sec:qpvswapsecure}, extension of the security analysis to the loss-tolerant setting in Section~\ref{sec:qpvswaploss}, and an upper bound to the entanglement required for an attack in Section~\ref{sec:qpvswapattack}. In Section~\ref{sec:exp_analysis} we analyse our protocol under realistic experimental conditions. Finally, in Section~\ref{sec:simulationres} we show the results of simulating our protocol under experimental conditions.

\section{Preliminaries}
\label{sec:prelim}\subsection{Notation}
We denote parties in QPV protocols by letters \textsf{A}, \textsf{B}, etc. and their quantum registers as $A_1 \cdots A_n$, $B_1 \cdots B_n$ and so on, respectively. Sometimes we may refer to ``all registers party $\mathsf{X}$ holds'' just by \textsf{X}, giving expression like $\operatorname{Pos}(\mathsf{A} \otimes \mathsf{B})$, for example. Cumulative distribution functions are written as $F_X$, where $X$ is either a random variable or explicitly the distribution. Unless otherwise indicated, $\lVert \cdot \rVert_p$ is the usual $p$-norm. Partial transposition of an operator $P$ with respect to party $\mathsf{B}$ is denoted $P^{T_\mathsf{B}}$. The set of PPT-measurements\footnote{I.e.\ sets of positive semi-definite operators adding up to the identity, whose partial transposes are positive semi-definite as well.} on two subsystems held by parties \textsf{A} and \textsf{B}, respectively, is PPT$(\mathsf{A}:\mathsf{B})$.

\subsection{The SWAP test}
The SWAP test was first introduced in \cite{buhrman_quantum_2001} for quantum fingerprinting as a useful tool to determine if two unknown states are identical or not. The quantum circuit of it is depicted in Figure~\ref{fig:swaptest}.

\begin{figure}[ht]
	\centering
	\includegraphics[width=0.5\textwidth]{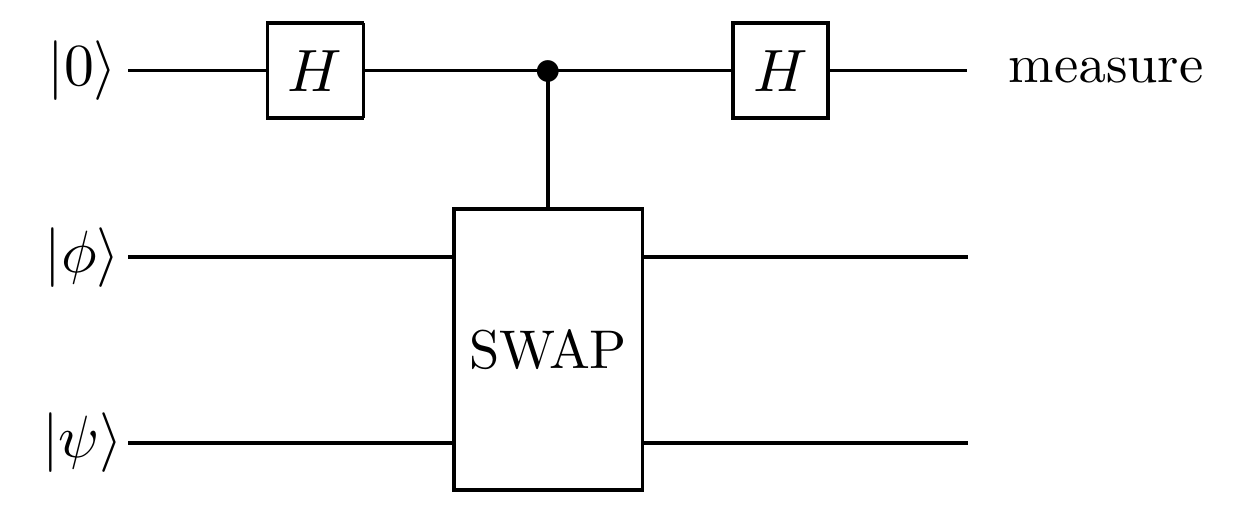}
	\caption{The SWAP test, taken from \cite{buhrman_quantum_2001}. $H$ denotes the Hadamard gate.}
	\label{fig:swaptest}
\end{figure}

\noindent The state to be measured in the computational basis is
\begin{align}
    (H \otimes \mathbbm{1}) \text{c-SWAP} (H \otimes \mathbbm{1}) \ket{0}\ket{\phi}\ket{\psi} = \frac{1}{2} \ket{0} (\ket{\phi}\ket{\psi} + \ket{\psi}\ket{\phi}) + \frac{1}{2} \ket{1} (\ket{\phi}\ket{\psi} - \ket{\psi}\ket{\phi}).
\end{align}
Therefore we have the measurement statistics
\begin{align}
    \Prob(0) = \frac{1+\lvert \braket{\psi|\phi}\rvert^2}{2} \qquad \text{and} \qquad \Prob(1) = \frac{1-\lvert\braket{\psi|\phi}\rvert^2}{2}.
    \label{eq: swap_distr}
\end{align}
The output distribution only depends on the overlap $\lvert \braket{\psi|\phi} \rvert$ between the input states. One notable special case is that for $\ket{\phi} = \ket{\psi}$ the SWAP operation has no effect and we get $\Prob(0) = 1$. Another advantage of the SWAP test is that it is easily implemented experimentally with a single beam splitter and two photon detectors \cite{jex2004comparing, garcia-escartin_swap_2013}. Its flexibility concerning input states and the simplicity of its experimental realization make it a good candidate for QPV.

\section{The \texorpdfstring{QPV$_{\textsf{SWAP}}$}{QPVswap} protocol}
\label{sec: qpvswap}
We define the protocol QPV$_{\textsf{SWAP}}(\beta_1, \dots, \beta_k)$, depicted in the space-time diagram in Figure \ref{fig:qpvswap}, as follows.

\begin{enumerate}
    \item By means of local and shared randomness or a secure private channel verifiers $\mathsf{V_A}$ and $\mathsf{V_B}$ uniformly draw a random overlap $\beta \in \{ \beta_1, \dots, \beta_k \}$ and agree on two uniformly random states $\ket{\psi}, \ket{\phi}$ such that $\lvert \braket{\psi|\phi} \rvert = \beta$. Then $\mathsf{V_A}$ prepares the state $\ket{\psi}$ and $\mathsf{V_B}$ prepares $\ket{\phi}$. Each verifier sends their state to $\mathsf{P}$ such that they arrive there simultaneously.
    \item The honest party $\mathsf{P}$ applies the SWAP test on the two quantum inputs as soon as they arrive at $\mathsf{P}$. This yields an output bit $z \in \{ 0, 1, \varnothing \}$, indicating $\mathsf{P}$'s measurement result or possibly a ``loss'' event. In particular, $\Prob(z=0\mid \beta, \text{not loss})=(1 + \beta^2)/2$ and $\Prob(z=1\mid \beta, \text{not loss})=(1 - \beta^2)/2$. Then $\mathsf{P}$ immediately sends $z$ to both verifiers $\mathsf{V_A}$ and $\mathsf{V_B}$.
    \item The verifiers closely monitor if they receive an answer in time and compare what they received. If they got different bits, or if at least one of their bits arrived too early/late, they abort and reject. Otherwise both verifiers add $z$ to their (ordered) lists of answers $L_\beta$.
    \item After having completed $R_\beta \geq R$ rounds with a conclusive answer $z \in \{0,1\}$, sequentially or in parallel for each $\beta$, they stop sending inputs, check if the rate of $\varnothing$ symbols is close enough\footnote{Say a rate $1-\eta$ is expected from \textsf{P}. The verifiers can apply an analogous statistical test around $1-\eta$ as described for the conclusive answers to check for any suspicious actions.} to what is expected from \textsf{P}, discard any rounds with answer $\varnothing$ and proceed to the statistical analysis on the sets of conclusive answers $C_\beta = L_\beta - \{ \varnothing \}$ for each $\beta$. They test if the sample $\hat{p}_\beta = \# \{ z \in C_\beta : z=0 \} / R_\beta$ on conclusive answers is contained in the $(1-\alpha)$-quantile around the expected $p_\beta = (1 + \beta^2)/2$.
    \item Only if they have received the same answer in time in every single round and if the statistical test was passed on all $L_\beta$, they accept. Otherwise, they reject.
\end{enumerate}

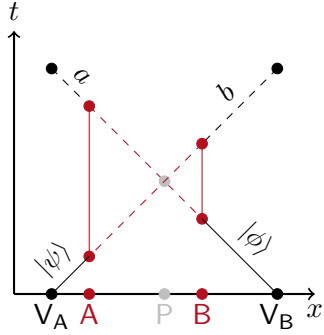
\begin{figure}[ht]
\centering
\begin{tikzpicture}
	\draw[->,black,thick] (-0.5,-3) -- (-0.5,0.5) node[above]{$t$};
	\draw[->,black,thick] (-0.5,-3) -- (3.5,-3) node[below]{$x$};
	
    \filldraw [black] (0,0) circle (2pt);
    \draw[black,dashed]  (0,0) -- (0.5,-0.5) node[midway, above, sloped]{$a$};
    \filldraw [QuSoft] (0.5,-0.5) circle (2pt);
    \draw[QuSoft,dashed]  (0.5,-0.5) -- (2,-2);
    \draw[QuSoft]  (0.5,-0.5) -- (0.5,-2.5);
    \filldraw [QuSoft] (0.5,-2.5) circle (2pt);
    \draw[black]  (0.5,-2.5) -- (0,-3) node[midway, above, sloped]{\small $\ket{\psi}$};
    \filldraw [black] (0,-3) circle (2pt) node[below]{$\mathsf{V_A}$};
    
    \filldraw [lightgray] (1.5,-1.5) circle (2pt);

    \filldraw [black] (3,0) circle (2pt);
    \draw[black,dashed]  (3,0) -- (2,-1) node[midway, above, sloped]{$b$};
    \filldraw [QuSoft] (2,-1) circle (2pt);
    \draw[QuSoft,dashed]  (2,-1) -- (0.5,-2.5);
    \draw[QuSoft]  (2,-1) -- (2,-2);
    \filldraw [QuSoft] (2,-2) circle (2pt);
    \draw[black]  (2,-2) -- (3,-3) node[midway, above, sloped]{\small $\ket{\phi}$};
    \filldraw [black] (3,-3) circle (2pt) node[below]{$\mathsf{V_B}$};
    
    \filldraw [QuSoft] (0.5,-3) circle (2pt) node[below]{$\mathsf{A}$};
    \filldraw [lightgray] (1.5,-3) circle (2pt) node[below]{$\mathsf{P}$};
    \filldraw [QuSoft] (2,-3) circle (2pt) node[below]{$\mathsf{B}$};
\end{tikzpicture}
\caption{Space-time diagram of the QPV$_\textsf{SWAP}$ protocol. We assume all information, quantum (---) and classical (-\,-\,-), travels at the speed of light. For graphical simplicity we have put $\mathsf{P}$ exactly in the middle of $\mathsf{V_A}$ and $\mathsf{V_B}$ (which is not necessary for the purposes of QPV). The attackers, not being at position $\mathsf{P}$, would like to convince the verifiers that they are at $\mathsf{P}$. Note that to have any chance of winning, attackers need to produce $a=b$.}
\label{fig:qpvswap}
\end{figure}
\noindent
Note that in essence the task in this protocol is to estimate the overlap $\beta$ of the input states. This is independent of the dimensionality/nature of the input states, making the protocol very flexible. To attack this protocol, it is evident that there need to be at least two attackers due to the timing constraint. A coalition of attackers has to position at least one party $\mathsf{A}$ between $\mathsf{V_A}$ and $\mathsf{P}$ and one party $\mathsf{B}$ between $\mathsf{P}$ and $\mathsf{V_B}$. Since the SWAP test is a joint operation on two quantum states, spatially separated attackers cannot apply the SWAP test, unless they have access to pre-shared entanglement, as we will show. Since any QPV protocol can be perfectly attacked if the attackers have access to enough pre-shared entanglement \cite{buhrman_position-based_2011}, we assume that the attackers have no pre-shared entanglement and we also restrict them to classical communication only for now.

To assess the security of this protocol, we consider the following. As the individual rounds are independent, the subsets $L_\beta$ of answers given input $\rho_\beta$ will be samples of a binomial distribution with parameters $R_\beta$ and some $q_{\beta}$\footnote{Here and in the following the parameter describes the fraction of ``0'' answers and we abbreviate $q_\beta(0)=q_\beta$}. The verifiers can then test if what they received matches closely enough with what they expect from an honest party. We define the statistical test to be done by the verifiers as follows:
\begin{enumerate}
    \item For each overlap $\beta$, they calculate the $(1-\alpha)$-quantile\footnote{In order to capture \textsf{P} with high probability, $\alpha$ can be set to a small number, e.g.\ $10^{-6}$.} around the ideal $p_\beta = (1+\beta^2)/2$, which gives a lower and an upper bound
    \begin{align}\label{equ:accreg_ideal}
    \begin{split}
        L_{\alpha, \beta} &\coloneqq z_{\frac{\alpha}{2}}(\beta, R_\beta) / R_\beta= F^{-1}_{\text{Bin}(R_\beta, p_\beta)}\left(\frac{\alpha}{2}\right) / R_\beta \\ U_{\alpha, \beta} &\coloneqq z_{1-\frac{\alpha}{2}}(\beta, R_\beta) / R_\beta= F^{-1}_{\text{Bin}(R_\beta, p_\beta)}\left(1-\frac{\alpha}{2} \right) / R_\beta,
    \end{split}
    \end{align}
    with $F^{-1}$ being the inverse cumulative distribution function. This defines an acceptance interval
    \begin{align}
        \mathsf{acc}_\beta(\alpha, R_\beta) \coloneqq [L_{\alpha, \beta}, U_{\alpha, \beta}].
    \end{align}
    \item For each overlap $\beta$, they check if the sample $\hat{p}_\beta \in \mathsf{acc}(\alpha, R_\beta)$. If this is the case for all $\beta$, they accept. Otherwise, they reject.
\end{enumerate}
By definition, the honest party will return a sample $\hat{p}^\mathsf{P}_\beta \in \mathsf{acc}_\beta(\alpha, R_\beta)$ with probability $1-\alpha$ and thus the test will accept \textsf{P} with high probability $(1-\alpha)^k = 1 - O(k\alpha)$. To optimize the overlap between their distribution and the acceptance regions, the attackers will attempt to respond as close to each $p_\beta$ as possible, with a binomial parameter of $p^\mathsf{AB}_\beta = p_\beta - \Delta_\beta$, defining a vector of differences 
\begin{align}
    \Delta = \begin{pmatrix}
		\Delta_{\beta_1} \\ \vdots \\ \Delta_{\beta_k}
	\end{pmatrix}.
\end{align}
If rounds are run in parallel, attackers could also decide to just respond a deterministic list with some fraction $\hat{p}^\mathsf{AB}$ of ``0'' answers. They could perfectly break the protocol if $\hat{p}^\mathsf{AB} \in \bigcap_\beta \mathsf{acc}_\beta(\alpha, R_\beta)$. This, however, we can always prevent by choosing the $R_\beta$'s large enough so that the acceptance regions for different overlaps get disjoint. Thus we need to evaluate
\begin{align}
    \Prob(\mathsf{acc} | \mathsf{attack}) \coloneqq \Prob \left( \hat{p}^\mathsf{AB}_\beta \in \mathsf{acc}_\beta(\alpha, R_\beta) \quad \forall \beta \right) = \prod_\beta \Prob \left( \hat{p}^\mathsf{AB}_\beta \in \mathsf{acc}_\beta(\alpha, R_\beta) \right) \eqqcolon \prod_\beta \Prob(\mathsf{acc}_\beta | \mathsf{attack}).
\end{align}
Now there are several cases to consider:
\begin{enumerate}
    \item[(1)]$\mathbf{\lVert \Delta \rVert_1 = 0.}$ Then $\Delta_\beta = 0$ for all $\beta$ and the attackers respond with the identical distribution as \textsf{P}, therefore $\Prob(\mathsf{acc} | \mathsf{attack}) = (1-\alpha)^k = 1 - O(k\alpha)$.
    \item[(2)] $\mathbf{p_{\boldsymbol{\beta}} \neq 1 \textbf{ and } \hat{p}^\mathsf{AB}_{\boldsymbol{\beta}} = 1}.$ Then $\Prob(\mathsf{acc} | \mathsf{attack}) =0$ as the $(1-\alpha)$-quantile around $p_\beta$ will exclude the value 1 (for sufficiently large $R_\beta$).
    \item[(3)] $\mathbf{p_{\boldsymbol{\beta}} = 1 \textbf{ and } \hat{p}^\mathsf{AB}_{\boldsymbol{\beta}} \neq 1}.$ Then $\Prob(\mathsf{acc}_\beta | \mathsf{attack}) = \left(p^\mathsf{AB}_\beta \right)^{R_\beta} = O\left(2^{-R_\beta}\right).$
    \item[(4)] $\mathbf{\lVert \Delta \rVert_1 \neq 0 \textbf{ and } p_{\boldsymbol{\beta}}, \hat{p}^\mathsf{AB}_{\boldsymbol{\beta}} \in \big[\frac{1}{2}, 1\big).}$ Then there exists a $\beta \in \{ \beta_1, \dots, \beta_k \}$ such that $\Delta_\beta \neq 0$. By using the Gaussian approximation for the binomial distributions (which we may apply as we can always make the number of rounds sufficiently large), one can show (cf. appendix \ref{AppendixExpSuppr}) that
    \begin{align}\label{equ:psuccattgen}
        \Prob(\mathsf{acc}_\beta | \mathsf{attack}) \lesssim \frac{\sqrt{2}f_{\beta}^\mathsf{AB}}{\sqrt{\pi R_\beta} \Delta_\beta} e^{-\left(\sqrt{R_\beta}\Delta_\beta - f_{\beta}^\mathsf{P} c_\alpha \right)^2 / \left(f_{\beta}^\mathsf{AB}\right)^2} = O \left( \frac{2^{-\Delta_\beta^2 R_\beta}}{\Delta_\beta \sqrt{R_\beta}}\right)
    \end{align}
    for functions $c_\alpha, f_\beta$ that are independent of $R_\beta$ and $\Delta_\beta$. Hence in this scenario the success probability of attackers is also exponentially suppressed.
\end{enumerate}
So unless $\Delta_\beta = 0$ for all $\beta$, we have exponential suppression in the attacker success probability $\Prob(\mathsf{acc} | \mathsf{attack})$. In the end, we can set a threshold $R_\text{threshold}$ for the number of rounds and the protocol is to be run until $R_\beta \geq R_\text{threshold}$ for all $\beta$. This will guarantee that any desired security level can be achieved by increasing $R_\text{threshold}$ ``uniformly'' over all $\beta$. We end up with a protocol that accepts an honest party with high probability and rejects (unentangled) attackers with high probability. A sketch is depicted in Figure~\ref{fig:acc_regions}.

\begin{figure}[ht]
	\centering
	\includegraphics[width=0.618\textwidth]{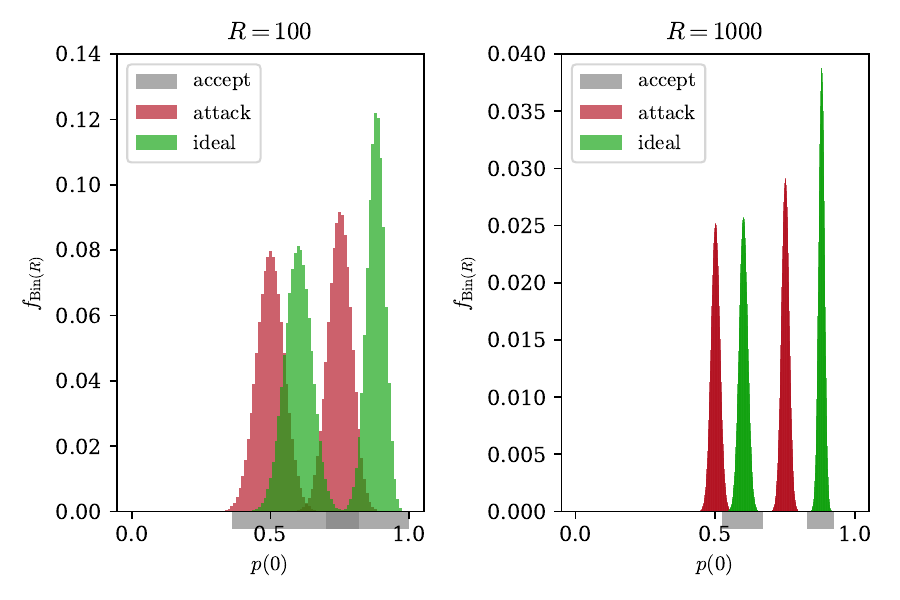}
	\caption{Sketch of the idea behind the statistical test. Acceptance regions around the expected honest $p_\beta$'s is defined such that \textsf{P} will be captured with high probability. Attackers trying to spoof verifiers by minimizing all $\Delta_\beta$ as well as possible have exponentially low (in $R$) probability of returning a sample contained in the acceptance regions for all $\beta$.}
	\label{fig:acc_regions}
\end{figure}

\noindent The analysis also suggests that optimally attackers want to minimize all $|\Delta_\beta|$ simultaneously\footnote{Note that minimizing $|p_\beta - p_\beta^{\mathsf{AB}}|$ also minimizes the Kullback-Leibler divergence $D_{\text{KL}}(P \parallel Q)$ between the corresponding binomial distributions $P$ and $Q$.}. We therefore choose to minimize $\lVert \Delta \rVert_1$. As LOCC $\subset$ PPT \cite{chitambar2014everything}, the following optimization program will provide a lower bound on $\lVert \Delta \rVert_1$ for LOCC restricted attackers. To account for imperfect quantum channel transmittance, we include a parameter $\eta \in (0, 1]$ and a third answer option $\varnothing$ (``loss''). Then $\Delta_\beta$ is to be evaluated conditioned on conclusive answers, i.e., $\Delta_\beta = p_\beta - \Tr[\Pi_0 \rho_\beta]/\eta$, where
\begin{align}
    \rho_\beta = \int_{\text{U}(2)} U \otimes U \ket{\psi \phi} \bra{\psi\phi} U^\dagger \otimes U^\dagger d\mu(U) = \frac{1}{3} \frac{1+\beta^2}{2} \Pi_\text{sym} + \frac{1-\beta^2}{2} \Pi_\text{a-sym}
\end{align}
is the mixed state the verifiers produce for overlap $\beta$ \cite{watrous_theory_2018}. Here $\Pi_\text{sym}, \Pi_\text{a-sym}$ are the projectors onto the symmetric and antisymmetric subspace, respectively, and $\mu$ is the Haar measure on the unitary group U$(2)$. These considerations lead us to the optimization:

\begin{align}\label{equ:opti}
\begin{split}
        \textbf{minimize: } & \lVert \Delta \rVert_1 \\
    \textbf{subject to: } &\Pi_0 + \Pi_1 + \Pi_\varnothing = \mathbbm{1}_{4} \\
    &\Pi_k \in \text{PPT}(\mathsf{A}:\mathsf{B}), \qquad k \in \{0,1, \varnothing \} \\
    &\Tr[\Pi_\varnothing \rho_\beta] = 1-\eta, \qquad \beta \in \{ \beta_1, \dots, \beta_k \}.
\end{split}
\end{align}
The constraints involving $\eta$ stem from the fact that \textsf{P} will produce the same inconclusive-rate on all overlaps $\beta$ and the attackers need to mimic that. An analogous statistical test with a $(1-\alpha)$-quantile around $\eta$ can be performed to check for this and it is clear that it is optimal for attackers to choose to reply inconclusive at the exact same rate as \textsf{P} would do for each $\beta$. The above program can be solved with conventional conic optimization libraries, e.g.\ \texttt{MOSEK} \cite{aps2020mosek}, and any example $\{ \beta_1, \dots, \beta_k \}$ we have tried yielded an optimal $\lVert \Delta \rVert_1 > 0$ independent of $\eta$, indicating the loss-tolerance of the protocol.

Security of QPV$_{\textsf{SWAP}}$ in the setting where attackers are allowed to use one round of simultaneous quantum communication follows from \cite{AllBuhSpeVer22}, as shown in Section~\ref{sec:qcsecurity}. We will now proceed and analyze the special case of overlaps $\{ 0, 1 \}$, i.e., sending orthogonal or identical states, in more detail and show analytically and numerically that it has desirable properties. 

\subsection{Security of \texorpdfstring{QPV$_{\textsf{SWAP}}(0,1)$}{QPVswap(0,1)} Protocol}
\label{sec:qpvswapsecure}

In this setting, there is the notion of a correct answer. On equal inputs the verifiers \textit{always} expect the answer `0'. This allows for a SDP formulation for maximizing the average success probability of identifying if the input states were equal/unequal. In appendix \ref{AppendixPsuccDelta} it is shown that the relation between the success probability $p_\text{succ}$ of correctly identifying equal/orthogonal and $\lVert \Delta \rVert_1$ is \begin{align}
    p_\text{succ} \leq u \implies \lVert \Delta \rVert_1 \geq \frac{3}{2} - 2u. 
\end{align} 
Having drawn the connection between $p_\text{succ}$ and $\lVert \Delta \rVert_1$, we will now proceed to show that there is a finite gap in the success probability of testing for equality between adversaries restricted to LOCC operations and an honest prover who can apply entangling measurements. Extending this single round protocol to $n$ rounds played in parallel, we will also show that the best strategy for adversaries is to simply apply the optimal single round strategy to every round individually, which shows \textit{strong parallel repetition} for QPV$_{\textsf{SWAP}}(0,1)$. Furthermore we show that in both cases there is no advantage for the attackers if they have the ability to declare loss on rounds, i.e.\ the probability of success conditioned on answering is independent of loss. The security of the protocol lies in the fact that an honest prover at his claimed position can apply entangling operations to the two incoming qubits and has a strictly higher probability of answering the question correctly than spatially separated adversaries who are restricted to single round LOCC operations. The fact that the protocol is also loss-tolerant comes from the fact that all inputs are quantum and there is no way to guess any classical information beforehand.

In general, the operation that has the highest probability of generating the correct answer is the SWAP test \cite{montanaro_survey_2018} and it gives a success probability $p_\text{succ}(\text{SWAP-test}) = 3/4$. We will show that the best strategy for LOCC adversaries gives at most a success probability of $p^\text{max}_\text{succ}(\text{LOCC}) = 2/3$. Since attackers return only a classical bit and they discard their post-measurement state, the most general type of measurement the attackers do is a $\textit{positive-operator-valued meaure}$ (POVM). The attackers' success probability for a given admissible POVM strategy $\Pi = \{ \Pi_0, \Pi_1 \}$ is then given by
\begin{align}
    p_\text{succ}(\Pi) := \frac{1}{2} \Tr[ \Pi_0 \rho_0 + \Pi_1 \rho_1 ].
    \label{avg_psucc}
\end{align}
Maximizing over all two-qubit LOCC measurements $\Pi^{\text{LOCC}}$ would give us the best probability of success of the attackers. However characterizing and maximizing over LOCC strategies is a mathematically complex task. We follow the method used in \cite{lim_loss-tolerant_2016}, and maximize our problem over the set of all positive partial transpose (PPT) operations. Since PPT measurements are a proper superset of LOCC measurements, any maximal success probability optimized over PPT measurements immediately upper bounds the success probability of all LOCC measurements. Furthermore, the PPT condition can be represented by a set of linear and positive semidefinite conditions \cite{cosentino_ppt-indistinguishable_2013} which enables us to write down the maximization problem as a semidefinite program (SDP) \cite{vandenberghe1996semidefinite}. This allows us to find exact solutions to the optimization problem if the values of the primal program and dual program coincide. In our case the SDP is as follows:

\begin{minipage}{0.48\textwidth}
\begin{align*}
    &\textbf{Primal Program}\\
    \textbf{maximize: } &\frac{1}{2} \Tr[ \Pi_0 \rho_0 + \Pi_1 \rho_1] \\
    \textbf{subject to: } &\Pi_0 + \Pi_1 = \mathbbm{1}_{2^2} \\
    & \Pi_k \in \text{PPT}(\mathsf{A}:\mathsf{B}), \ \ \ k \in \{0,1\} \\[1ex]
\end{align*}
\end{minipage}
\begin{minipage}{0.48\textwidth}
\begin{align*}
    &\textbf{Dual Program} \\
    \textbf{minimize: } &\Tr[Y] \\
    \textbf{subject to: } &Y - Q^{T_\mathsf{B}}_{i} - \rho_i / 2 \succeq 0, \ \ \ i \in \{0,1\} \\
    & Y\in \text{Herm}(\mathsf{A} \otimes \mathsf{B}) \\
    & Q_i \in \text{Pos}(\mathsf{A} \otimes \mathsf{B}), \ \ \ i \in \{0,1\}, \\
\end{align*}
\end{minipage}

\noindent
Note that the primal program implies a lower bound and the dual program an upper bound to $p^{\text{max}}_\text{succ}(\Pi^{\text{PPT}})$. We find an exact optimal solution to the SDP of 2/3 (see Appendix \ref{AppendixSingleRound}), giving an upper bound of the success probability optimized over all LOCC measurements of
\begin{align} \label{eqref:LOCC_upperbound}
    p^{\text{max}}_\text{succ}(\Pi^{\text{LOCC}}) \leq \frac{2}{3}.
\end{align}
The input states $\rho_0$ and $\rho_1$ have the exact same mixed state matrices as the result of uniformly choosing a mutually unbiased basis and sending either equal or orthogonal states (from the chosen basis) to \textsf{P}. This indicates an optimal LOCC strategy. Assume the incoming qubits are encoded in MUB $b$, and that the attackers choose a random MUB $b'$, measure both incoming qubits in the basis $b'$, send the measurement outcome to each other, and return equal if the measurement outcomes are equal and unequal otherwise. Then their probability of success is exactly $2/3$, since 
\begin{align}
    \mathbb{P}(\text{success}) = \mathbb{P}(b' = b)\mathbb{P}(\text{success} | b' = b) + \mathbb{P}(b' \neq b)\mathbb{P}(\text{success}| b' \neq b)
    = \frac{1}{3} \cdot 1 + \frac{2}{3} \cdot \frac{1}{2} = \frac{2}{3}.
\end{align}
This attack strategy uses only local measurements and a single round of communication, so it is a valid single round LOCC operation. Thus we find that the upper bound in \eqref{eqref:LOCC_upperbound} over LOCC measurements is in fact a tight bound attained by LOCC. 

We have shown that the probability of success for identifying if the given inputs were equal or not for the QPV$_{\textsf{SWAP}}(0,1)$ protocol is strictly lower for attackers restricted to LOCC measurements than for an honest verifier who can apply entangling operations (2/3 versus 3/4 respectively). Over sequential multi-round protocols, where we only perform a new run of the protocol after the previous is finished, the verifiers can increase the precision of detecting LOCC attackers to any limit they desire.

An important question to ask is whether we can extend the single round protocol to a general $n$-round parallel protocol, where the verifiers send $n$ qubits from both sides to form the density matrix $\rho_s = \rho_{s_0} \otimes \rho_{s_1} \otimes \dots \otimes \rho_{s_{n-1}}$ for $s \in \{0,1\}^n$. Note that this does not follow naively from the single round security proof since attackers can now in principle take blocks of inputs and apply entangling operations on them. We will prove that for the QPV$_{\textsf{SWAP}}$ protocol strong parallel repetition does indeed hold, i.e.\ the probability of success of winning $n$ rounds decreases as $(2/3)^n$, implying that the best strategy for attackers is to simply attack each round individually. Again we can write down the problem as a SDP optimization task where we optimize over all PPT operations on the $2n$ qubits the attackers receive. 

\begin{minipage}{0.48\textwidth}
\begin{align*}
    &\textbf{Primal Program}\\
    \textbf{maximize: } &\frac{1}{2^n} \sum_{s \in \{0,1\}^n} \Tr[ \Pi_{s} \rho_{s}] \\
    \textbf{subject to: } &\sum_{s \in \{0,1\}^n} \Pi_{s} = \mathbbm{1}_{2^{2n}} \\
    & \Pi_{s} \in \text{PPT}(\mathsf{A}:\mathsf{B}), \ \ \ s \in \{0,1\}^n\\
\end{align*}
\end{minipage}
\begin{minipage}{0.48\textwidth}
\begin{align*}
    &\textbf{Dual Program}\\[1ex]
    \textbf{minimize: } &\Tr[Y] \\[3ex]
    \textbf{subject to: } &Y - Q^{T_\mathsf{B}}_{s} - \rho_{s} / 2^n \succeq 0, \ \ \ s \in \{0,1\}^n \\
    & Y \in \text{Herm}(\mathsf{A} \otimes \mathsf{B}) \\
    & Q_{s} \in \text{Pos}(\mathsf{A} \otimes \mathsf{B}). \\
\end{align*}
\end{minipage}

\noindent
In Appendix \ref{AppendixParallel} we find an explicit analytical solution to the dual problem. The solution is non-trivial and depends on the specifics of the QPV$_{\textsf{SWAP}}$ protocol, so it does not generalize naturally to strong parallel repetition results for other protocols. The solution yields a value of $(2/3)^n$, which bounds the probability of success under LOCC measurements by $(2/3)^n$. A feasible solution to the primal problem is to fill in the single round solution $n$ times and this has success probability $(2/3)^n$. Since this strategy coincides with the previously mentioned single round LOCC measurement applied to each of the individual rounds of $\rho_{s_i}$, we find that the upper bound of $(2/3)^n$ is again attained by an LOCC measurement and tight. Thus we show strong parallel repetition for the  QPV$_{\textsf{SWAP}}$ protocol against attackers restricted to LOCC operations. 

Strong parallel repetition is a useful result for the practical implementation of QPV protocols. First of all it implies that when playing multiple rounds we don't have to wait until a single round is finished, thus simplifying the timing constraints of multiple rounds. Secondly, it implies a linear lower bound on the entanglement adversaries need to attack the protocol perfectly as shown in section \ref{sec:qpvswapattack}.

\subsection{Loss-Tolerance of \texorpdfstring{QPV$^n_{\textsf{SWAP}}$}{repeated QPVswap} Protocol}
\label{sec:qpvswaploss}
In the previous section we have shown that the QPV$_{\textsf{SWAP}}$ protocol is secure against attackers restricted to LOCC attackers in the case where attackers have to answer in every round. However, in practice an honest prover will only answer on a fraction of the rounds played due to channel loss and imperfect measurements. In order to prove security against any coalition of attackers in the setting with channel loss, we must assume that attackers will never suffer any loss when they attack a protocol\footnote{They could position themselves very close to the verifiers and have perfect communication channels, for example.}. When classical information is sent, such as in the QPV$_{\text{BB84}}$ protocol \cite{kent_quantum_2011, buhrman_position-based_2011}, attackers may guess the classical information that is being sent. If they guess incorrectly they discard the round and declare a loss ($\varnothing$), if they guess correctly they can continue and successfully attack the protocol since the classical information is known to both attackers after communication. If the loss rate is high enough, attackers can hide their incorrect guesses in the loss declarations and the verifiers cannot distinguish the attackers from an honest prover. Note that in order to pretend a loss without being detected, attackers must declare a loss with equal probability on every input. To prove loss tolerance, we can incorporate loss in the SDP setting and show that the optimal solution of the SDP is independent of the loss, similar to the method in \cite{lim_loss-tolerant_2016}.

We can relatively straightforwardly add the condition that attackers must mimic a certain loss rate $(1 - \eta)$ on all inputs. We first show that in the parallel repetition case $p_\text{succ}$  is independent of $\eta$ when attackers either answer conclusively on all inputs or don't answer at all. In the following proposition, we use this simple result to show that this property implies that $p_\text{succ}$  stays independent of $\eta$ when declaring a loss on any \textit{subset} of rounds is allowed.
\begin{proposition}\label{prop:lossy_swap_allornothing}
Any multi-round QPV protocol that fulfills strong parallel repetition security against adversaries restricted to LOCC operations and is tolerant against declaring loss on all $n$ rounds, is also tolerant against declaring loss on any subset of rounds.
\end{proposition}

\begin{proof}
Suppose we have a secure $n$-round QPV protocol with strong parallel repetition. Then the $n$-round success probability for attackers is $p_{n} = p_1^n$ for some single round probability $p_1$. Suppose we perform $n$ rounds and we allow adversaries to only answer on $k$ rounds and to declare a loss on the remaining $(n-k)$ rounds, and suppose that there is some attacking strategy $S$ restricted to LOBC measurements that has a probability $p_S > p_1^k$ of being correct on this subset. We will show that this leads to a contradiction. Consider a protocol like the $k$-round protocol QPV$^k_{\textsf{SWAP}}$, which is secure and loss tolerant on all rounds by assumption and has success probability $p_k = p_1^k$. Since individual rounds are product states, attackers may create $n-k$ independent extra rounds locally of which they can forget the answer. This creates a $n$-round protocol. The attackers can now apply their strategy $S$. With probability $1/\binom{n}{k}$ they get an answer on their initial $k$ rounds that is correct with success probability $p_S$. And with probability $1 - 1/\binom{n}{k}$ they receive the wrong subset of $k$ rounds, in which case the attackers declare a loss (on all rounds). This defines an LOCC attack with a conditional winning probability $p_S > p_1^k$ and loss rate of $1 - 1/\binom{n}{k}$, which contradicts our assumption that the maximal success probability of being correct on the $k$-round protocol is $p_1^k$ for any loss. Therefore, for any subset of $k$ rounds out of the total of $n$ rounds, the maximal success probability $p_k$ on this subset is $p_1^k$.
\end{proof}

\noindent Next, we formulate an SDP to maximize the probability of success conditioned on a conclusive answer $p_\text{succ}^{\text{max}}(n, \eta)$ in the $n$-round parallel repetition case ($n=1$ corresponds to the single round protocol). 

\begin{minipage}{0.45\textwidth}
\begin{align*}
    &\textbf{Primal Program}\\
    \textbf{maximize: } &\frac{1}{2^n \eta} \sum_{s \in \{0,1\}^n} \Tr[ \tilde{\Pi}_{s} \rho_{s}] \\
    \textbf{subject to: } & \left(\sum_{s \in \{0,1\}^n} \tilde{\Pi}_{s} \right) + \tilde{\Pi}_\varnothing = \mathbbm{1}_{2^{2n}} \\
    & \Tr[\tilde{\Pi}_\varnothing \rho_s] = 1 - \eta, \ \ \ s \in \{0,1\}^n \\
    & \tilde{\Pi}_{s} \in \text{PPT}(\mathsf{A}:\mathsf{B}), \ \ \ s \in \{0,1\}^n \cup \varnothing \\
\end{align*}
\end{minipage}
\begin{minipage}{0.45\textwidth}
\begin{align*}
    &\textbf{Dual Program}\\
    \textbf{minimize: } & \frac{\Tr[\tilde{Y}] - (1 - \eta) \gamma}{\eta} \\
    \textbf{subject to: } &\tilde{Y} - \tilde{Q}^{T_\mathsf{B}}_{s} - \rho_{s} / 2^n \succeq 0, \ \ \ s \in \{0,1\}^n \\
    & 2^{2n} ( \tilde{Y} - \tilde{Q}^{T_\mathsf{B}}_\varnothing ) - \gamma \mathbbm{1}_{2^{2n}} \succeq 0 \\
    & \tilde{Y} \in \text{Herm}(\mathsf{A} \otimes \mathsf{B}) \\
    & \tilde{Q}_{s} \in \text{Pos}(\mathsf{A} \otimes \mathsf{B}), \ \ \ s \in \{0,1\}^n \cup \varnothing \\
    & \gamma \in \mathbb{R}. \\
\end{align*}
\end{minipage}
\noindent
From the analysis in Appendix \ref{AppendixLossTolerantParallelRepetition}, we see that the solution of the SDP is again $(2/3)^n$, independent of $\eta$, upper bounding the attackers restricted to LOCC measurements. The strategy in which attackers apply with probability $\eta$ the regular $n$-round parallel repetition attack and with probability $(1-\eta)$ discard everything again has conditional success probability $(2/3)^n$ so the bound is tight. By Proposition~\ref{prop:lossy_swap_allornothing}, we have that QPV$^n_{\textsf{SWAP}}$ is tolerant against loss on any subset of rounds, establishing full loss tolerance.

\subsection{Security with quantum communication}\label{sec:qcsecurity}


We now show that, no matter the transmission rate $\eta$, the SWAP-test cannot be perfectly simulated by unentangled attackers even if they have access to quantum communication. Our argument relies on the same fact for the Bell measurement, as is proven in \cite{AllBuhSpeVer22}. 

Abstractly, the SWAP-test implements the POVM $\{ \Pi_\text{sym}, \Pi_\text{a-sym} \}$ of projecting onto either the symmetric or the anti-symmetric subspace. In particular, it allows to perfectly distinguish $\ket{\Psi_-}$ from the other Bell states, in particular this shows that there is a finite gap between the best quantum communication attack without loss. With loss we will show that if the SWAP-test could be implemented perfectly with local actions and one round of quantum communication for some $0 <\eta \leq 1$, then so could the Bell measurement with some different $\eta' < \eta$, contradicting our result in \cite{AllBuhSpeVer22}, thus implying that there exist no perfect lossy attack of the SWAP-test.

\begin{proposition}
    QPV$_\textsf{SWAP}$ cannot be perfectly attacked if attackers $A$, $B$ can use quantum communication between them, no matter the loss rate $1-\eta$. That is, $\lVert \Delta \rVert_1 > 0$ for any $\eta \in (0,1]$.
\end{proposition}

\begin{proof}
    Assume there is a procedure, using only local actions and one round of simultaneous quantum communication, perfectly simulating $\{ \Pi_\text{sym}, \Pi_\text{a-sym} \}$ with probability $0 <\eta \leq 1$. Then, conditioned on their procedure giving a conclusive result (which happens with probability $\eta$), $\mathsf{A}, \mathsf{B}$ could do the following in QPV$_\textsf{Bell}$ with a Bell measurement at \textsf{P} and an input chosen uniformly at random from $\{ \ket{\Phi_+}, \ket{\Phi_-}, \ket{\Psi_+}, \ket{\Psi_-}\}$:
\begin{itemize}
    \item Whenever their procedure returns ``anti-symmetric'', return $\ket{\Psi_-}$
    \item Whenever it returns ``symmetric'', return the loss symbol $\varnothing$
\end{itemize}
However, this would be suspicious, because the only conclusive answers would be $\ket{\Psi_-}$. In order to achieve $\Prob(\varnothing \, | \, B_i) = 1-\eta$ for all Bell states $\ket{B_i}$ and $\Prob(B_i \, | \, \text{concl.}) = 1/4$, as the honest \textsf{P} would do in QPV$_\textsf{Bell}$, they could apply $\mathds{1}_A \otimes (X^a Z^b)_B$ with $a,b \in \{ 0, 1 \}$ chosen uniformly at random in each round as soon as they receive the inputs. This just transfers the input to a different Bell state. If they adjust their responses to
\begin{itemize}
    \item Whenever their procedure returns ``anti-symmetric'', answer $\mathds{1}_A \otimes (Z^b X^a)_B \ket{\Psi_-}$
    \item Whenever it returns ``symmetric'', answer the loss symbol $\varnothing$
\end{itemize}
They achieve $\Prob(\varnothing \, | \, B_i) = 1-\eta$ as well as $\Prob(B_i \, | \, \text{concl.}) = 1/4$ and whenever they do answer conclusively, they will be correct (by assumption). But this would give them a perfect attack on QPV$_\textsf{Bell}$ with some probability $\eta' < \eta$ (because they throw away the ``symmetric'' measurement results). This contradicts the the fact that $p_\text{succ}(\eta) < 1$ for all $\eta$ in QPV$_\textsf{Bell}$.
\end{proof}

\subsection{Entanglement attack}
\label{sec:qpvswapattack}
It turns out that there is a perfect attack on QPV$_{\textsf{SWAP}}(\beta_1, \dots, \beta_k)$ using one pre-shared maximally entangled state between the attackers. This gets apparent if one looks at the \textit{purified} version of the protocol. In this setting, the attackers do not receive mixed states from the verifiers, but rather halves of the corresponding purification. In QPV$_{\textsf{SWAP}}(\beta_1, \dots, \beta_k)$ this purification is a maximally entangled state on each side. This does not change anything about the input/output distributions of the attackers, as already noted in \cite{buhrman_position-based_2011}. Entanglement swapping is captured in the identity
\begin{align}
    \ket{\Phi_+}_{12}\ket{\Phi_+}_{34} = \frac{1}{2} \left( \ket{\Phi_+}_{14}\ket{\Phi_+}_{23} + \ket{\Phi_-}_{14}\ket{\Phi_-}_{23} + \ket{\Psi_+}_{14}\ket{\Psi_+}_{23} + \ket{\Psi_-}_{14}\ket{\Psi_-}_{23} \right).
    \label{equ:eswap}
\end{align}
Applying a Bell state measurement (BSM) on registers (23) swaps entanglement from registers (12) and (34) to (14) and (23). The intriguing aspect is that registers (14) could be causally separated, yet are entangled after the BSM. We use this to prove the following statement.

\begin{theorem}
The \emph{SWAP-}test can be perfectly simulated using one pre-shared maximally entangled state and one round of classical communication between $\mathsf{A}$ and $\mathsf{B}$. Thus $n$ pre-shared EPR pairs are sufficient to attack \emph{QPV}$_\mathsf{SWAP}^n(\beta_1, \dots, \beta_k)$, and $\sim 0.103 n$ pre-shared EPR pairs are necessary to attack \emph{QPV}$_\mathsf{SWAP}^n(0, 1)$.
\end{theorem}
\begin{proof}
In the purified protocol, verifiers do not send out their respective mixed states $\rho_{V_0}, \rho_{V_1}$ but rather halves of the corresponding purified pure state. Note that for a given list of overlaps $\mathcal{O} = \{ \beta_1, \dots, \beta_k \}$ the total input state is $\rho = \frac{1}{k} \sum_{\beta \in \mathcal{O}} \rho_\beta$. One can then calculate
\begin{align*}
    \Tr_{V_0}[\rho] = \Tr_{V_1}[\rho] = \frac{\mathbbm{1}_2}{2}
\end{align*}
and therefore the purifications at $\mathsf{V_0}, \mathsf{V_1}$ are maximally entangled states, say $\ket{\Phi_+}$. The resulting entanglement structure between the verifiers and attackers throughout the attack is as depicted in Figure~\ref{fig:eattack}.
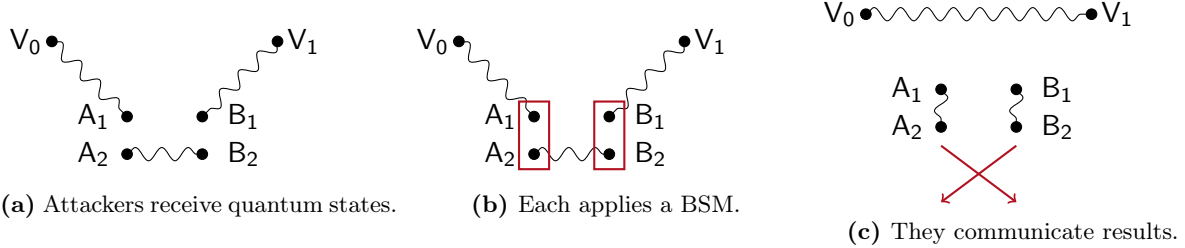
\begin{figure}[ht]
\centering
\begin{subfigure}{0.32\textwidth}
    \begin{tikzpicture}
        \filldraw [black] (0,0) circle (2pt) node[left] {$\mathsf{V_0}$};
        \draw[black] [,decorate,decoration=snake] (0,0) -- (1,-1);
        \filldraw [black] (1,-1) circle (2pt) node[left] {$\mathsf{A_1} \hspace{1mm}$};
    
        \filldraw [black] (3,0) circle (2pt) node[right] {$\mathsf{V_1}$};
        \draw[black] [,decorate,decoration=snake] (3,0) -- (2,-1);
        \filldraw [black] (2,-1) circle (2pt) node[right] {$\hspace{2mm} \mathsf{B_1}$};
    
        \filldraw [black] (1,-1.5) circle (2pt) node[left] {$\mathsf{A_2}\hspace{1mm}$};
        \draw[black] [,decorate,decoration=snake] (1,-1.5) -- (2,-1.5);
        \filldraw [black] (2,-1.5) circle (2pt) node[right] {$\hspace{2mm}\mathsf{B_2}$};
    \end{tikzpicture}
    \caption{Attackers receive quantum states.}
    \label{fig:Entangl_AB_start}
\end{subfigure}
\begin{subfigure}{0.32\textwidth}
    \begin{tikzpicture}
        \filldraw [black] (0,0) circle (2pt) node[left] {$\mathsf{V_0}$};
        \draw[black] [,decorate,decoration=snake] (0,0) -- (1,-1);
        \filldraw [black] (1,-1) circle (2pt) node[left] {$\mathsf{A_1} \hspace{1mm}$};
    
        \filldraw [black] (3,0) circle (2pt) node[right] {$\mathsf{V_1}$};
        \draw[black] [,decorate,decoration=snake] (3,0) -- (2,-1);
        \filldraw [black] (2,-1) circle (2pt) node[right] {$\hspace{2mm} \mathsf{B_1}$};
    
        \filldraw [black] (1,-1.5) circle (2pt) node[left] {$\mathsf{A_2}\hspace{1mm}$};
        \draw[black] [,decorate,decoration=snake] (1,-1.5) -- (2,-1.5);
        \filldraw [black] (2,-1.5) circle (2pt) node[right] {$\hspace{2mm}\mathsf{B_2}$};
        
        \draw[QuSoft, thick] (0.8,-0.8) rectangle (1.2,-1.7);
        \draw[QuSoft, thick] (1.8,-0.8) rectangle (2.2,-1.7);
    \end{tikzpicture}
    \caption{Each applies a BSM.}
    \label{fig:Entangl_AB_meas}
\end{subfigure}
\begin{subfigure}{0.32\textwidth}
    \begin{tikzpicture}
        \filldraw [black] (0,0) circle (2pt) node[left] {$\mathsf{V_0}$};
        \draw[black] [,decorate,decoration=snake] (0,0) -- (3,0);
        \filldraw [black] (1,-1) circle (2pt) node[left] {$\mathsf{A_1} \hspace{1mm}$};
    
        \filldraw [black] (3,0) circle (2pt) node[right] {$\mathsf{V_1}$};
        \draw[black] [,decorate,decoration=snake] (1,-1) -- (1,-1.5);
        \filldraw [black] (2,-1) circle (2pt) node[right] {$\hspace{2mm} \mathsf{B_1}$};
    
        \filldraw [black] (1,-1.5) circle (2pt) node[left] {$\mathsf{A_2}\hspace{1mm}$};
        \draw[black] [,decorate,decoration=snake] (2,-1) -- (2,-1.5);
        \filldraw [black] (2,-1.5) circle (2pt) node[right] {$\hspace{2mm}\mathsf{B_2}$};
    
        \draw[->, QuSoft, thick] (1,-1.75) -- (2,-2.5);
        \draw[->, QuSoft, thick] (2,-1.75) -- (1,-2.5);
    \end{tikzpicture}
    \caption{They communicate results.}
    \label{fig:Entangl_AB_comm}
\end{subfigure}
\caption{Entanglement structure throughout a purified protocol.}
\label{fig:eattack}
\end{figure}
Assume now that attackers pre-share one EPR pair, say also $\ket{\Phi_+}$, in registers $A_2B_2$. Then the total input state of the protocol can be rewritten as
\begin{align*}
    \ket{\Phi_+}_{V_0A_1} \ket{\Phi_+}_{V_1B_1} \ket{\Phi_+}_{A_2B_2} = &\frac{1}{2} \bigg[ \ket{\Phi_+}_{V_0V_1} \otimes \frac{1}{2} \Big( \ket{\Phi_+}\ket{\Phi_+} + \ket{\Phi_-}\ket{\Phi_-} + \ket{\Psi_+}\ket{\Psi_+} + \ket{\Psi_-}\ket{\Psi_-} \Big)_{A_1A_2B_1B_2} \\
    &+ \ket{\Phi_-}_{V_0V_1} \otimes \frac{1}{2} \Big( \ket{\Phi_+}\ket{\Phi_-} + \ket{\Phi_-}\ket{\Phi_+} + \ket{\Psi_+}\ket{\Psi_-} + \ket{\Psi_-}\ket{\Psi_+} \Big)_{A_1A_2B_1B_2} \\
    &+ \ket{\Psi_+}_{V_0V_1} \otimes \frac{1}{2} \Big( \ket{\Phi_+}\ket{\Psi_+} - \ket{\Phi_-}\ket{\Psi_-} + \ket{\Psi_+}\ket{\Phi_+} - \ket{\Psi_-}\ket{\Phi_-} \Big)_{A_1A_2B_1B_2} \\
    &+ \ket{\Psi_-}_{V_0V_1} \otimes \frac{1}{2} \Big(  \ket{\Phi_-}\ket{\Psi_+} - \ket{\Psi_+}\ket{\Phi_-} + \ket{\Psi_-}\ket{\Phi_+} - \ket{\Phi_+}\ket{\Psi_-}\Big)_{A_1A_2B_1B_2} \bigg].
\end{align*}
By separately performing a BSM on their two respective qubits in $A_1A_2$ and $B_1B_2$, the attackers will get one of the 16 measurement result combinations in the above equation and collapse the state in $V_0V_1$ into the corresponding Bell state. By communicating their results (classically) to each other, they can uniquely identify the state in the verifiers' registers $V_0V_1$. Let them then use the following strategy: answer `0', whenever they infer that the verifiers hold a symmetric state and `1', whenever it is an anti-symmetric state. Doing so, attackers effectively perform a measurement $\{ \Pi_\text{sym}, \Pi_\text{a-sym} \}$ on $V_0V_1$, which is precisely what the SWAP test does. To make this argument formal, let \textsf{A} and \textsf{B} receive registers $V_0$ and $V_1$, respectively. Inspired by the above, define
\begin{alignat*}{2}
    W &\coloneqq &&\mathbbm{1}_{V_0} \otimes \text{SWAP}_{A_2 V_1} \otimes \mathbbm{1}_{B_2}\\
    \Pi_0 &\coloneqq \Big( &&\ketbra{\Phi_+ \Phi_+}{\Phi_+ \Phi_+} + \ketbra{\Phi_- \Phi_-}{\Phi_- \Phi_-} + \ketbra{\Psi_+ \Psi_+}{\Psi_+ \Psi_+} + \ketbra{\Psi_- \Psi_-}{\Psi_- \Psi_-} + \\
    &{} &&\ketbra{\Phi_+ \Phi_-}{\Phi_+ \Phi_-} + \ketbra{\Phi_- \Phi_+}{\Phi_- \Phi_+} + \ketbra{\Psi_+ \Psi_-}{\Psi_+ \Psi_-} + \ketbra{\Psi_- \Psi_+}{\Psi_- \Psi_+} +\\
    &{} &&\ketbra{\Phi_+ \Psi_+}{\Phi_+ \Psi_+} + \ketbra{\Phi_- \Psi_-}{\Phi_- \Psi_-} + \ketbra{\Psi_+ \Phi_+}{\Psi_+ \Phi_+} + \ketbra{\Psi_- \Phi_-}{\Psi_- \Phi_-} \Big)_{V_0A_2V_1B_2} \\
    \Pi_1 &\coloneqq &&\Big(\ketbra{\Phi_- \Psi_+}{\Phi_- \Psi_+} + \ketbra{\Psi_+ \Phi_-}{\Psi_+ \Phi_-} + \ketbra{\Psi_- \Phi_+}{\Psi_- \Phi_+} + \ketbra{\Phi_+ \Psi_-}{\Phi_+ \Psi_-} \Big)_{V_0A_2V_1B_2}.
\end{alignat*}
Note that $\Pi_0$ and $\Pi_1$ are positive semi-definite and $\Pi_0 + \Pi_1 = \mathbbm{1}$, so that $\{ \Pi_0, \Pi_1 \}$ is a valid POVM. It can then be explicitly checked that
\begin{align}
    \begin{split}
        \Pi_\text{sym}^{V_0V_1} \rho_\beta^{V_0V_1} &= \Tr_{A_2B_2}\Big[ W \Pi_0^{V_0A_2V_1B_2} W^\dagger \,\, \left(\rho_\beta^{V_0V_1} \otimes \ketbra{\Phi_+}{\Phi_+}^{A_2B_2} \right)\Big] \\
        \Pi_\text{a-sym}^{V_0V_1} \rho_\beta^{V_0V_1} &= \Tr_{A_2B_2}\Big[ W \Pi_1^{V_0A_2V_1B_2} W^\dagger \,\, \left(\rho_\beta^{V_0V_1} \otimes \ketbra{\Phi_+}{\Phi_+}^{A_2B_2} \right)\Big],
    \end{split}
    \label{eq:swap_eattack}
\end{align}
for any $\beta \in [0,1]$. Equation \eqref{eq:swap_eattack} shows that the above attack using one pre-shared EPR pair is exactly the same as the honest action on $\rho_\beta$. In other words, the attackers can perfectly simulate the SWAP test and thus reproduce \textsf{P}. To get a lower bound on the required entanglement resource in order to break QPV$_\mathsf{SWAP}(0, 1)$ we can use an argument already mentioned in Lemma V.3 in \cite{beigi_simplified_2011}. It says that if the attackers pre-share a $d$-dimensional resource state $\tau_{\mathsf{AB}}$ then the success probability (of the attackers achieving that the verifiers accept them) is related to the success probability without a pre-shared resource in the following way:
\begin{align}
    p_{\text{succ}|\tau_{\mathsf{AB}}} \leq d p_{\text{succ}|\emptyset}.
\end{align}
For the $\beta \in \{0, 1\}$ case, we argued at the beginning of section~\ref{sec:qpvswapsecure} that the optimal attack strategy is to produce no error on orthogonal inputs and then accept whatever error that means for the identical inputs (which turned out to be $\Delta_1 = 1/4$). Then, the probability that the verifiers accept attackers is basically $(3/4)^{n_=}$, where $n_=$ is the number of rounds with identical inputs. Since each overlap is chosen with probability $1/2$ in each round, we have for $n$ rounds that $\mathbb{E}[n_=] = n/2$. Hence, we expect $p_{n, \text{succ}|\tau_{\mathsf{AB}}} \leq d \left( \frac{3}{4} \right)^{n/2}$ and thus
\begin{align}
    p_{n, \text{succ}|\tau_{\mathsf{AB}}} < 1 \qquad \text{as long as} \qquad d < \left( \frac{4}{3} \right)^{n/2}.
\end{align}
If $m$ is the number of EPR pairs in $\tau_{\mathsf{AB}}$, so that $d = 2^{2m}$, it follows that, in expectation,
\begin{align}
    p_{n, \text{succ}|\tau_{\mathsf{AB}}} < 1 \qquad \text{as long as} \qquad m < \frac{1}{4} \log \left( \frac{4}{3} \right) n \approx 0.103n.
\end{align}
\end{proof}

\section{\texorpdfstring{QPV$_{\textsf{SWAP}}$}{QPVswapexp} with realistic experimental conditions}
\label{sec:exp_analysis}
\subsection{Practical considerations}
	The SWAP-test has been shown to be equivalent to the Hong-Ou-Mandel (HOM) interference measurement \cite{Hong-Ou-Mandel} with just one $50/50$ beam splitter and two photon detectors \cite{jex2004comparing, garcia-escartin_swap_2013}. We call this the \textbf{BS} setup, as only a single beam splitter is used. If the photons bunch into one detector arm, the answer shall be ``0'', if both detectors register a click it shall be ``1''. However, for click/no-click detectors there is a problem with this simple setup, as signal loss can convert ``1'' answers to ``0'' answers. For high loss rates one would always get $p_\beta(0) \approx 1$, irrespective of the overlap and even without further equipment errors because most of the time only one state will arrive. Hence the \textbf{BS} setup will be insecure unless one uses number-resolution (NR) detectors. With these, single clicks at one detector get filtered out instead of delivering a wrong answer. NR detectors also filter out $k>2$ click events so that the ideal SWAP-test distribution of $p_\beta(0) = \frac{1 + \beta^2}{2}$ is fairly well preserved\footnote{With a generic $(R,T)$ beamsplitter, one has $p_\beta(0) = 4 \lvert R \rvert^2 \lvert T \rvert^2 \frac{1 + \beta^2}{2}$}, even with experimental errors. Creating true NR detectors is an active field of research, but at the moment they are still at early stage and somewhat hard to operate \cite{curtis2021single, Endo_2021}. We therefore use two further beam splitters and four click/no-click detectors to achieve probabilistic NR. We call this the \textbf{3BS} setup, as depicted in figure~\ref{fig:setups}.
	\begin{figure}[h]
		\centering
		\includegraphics[width=0.66\textwidth]{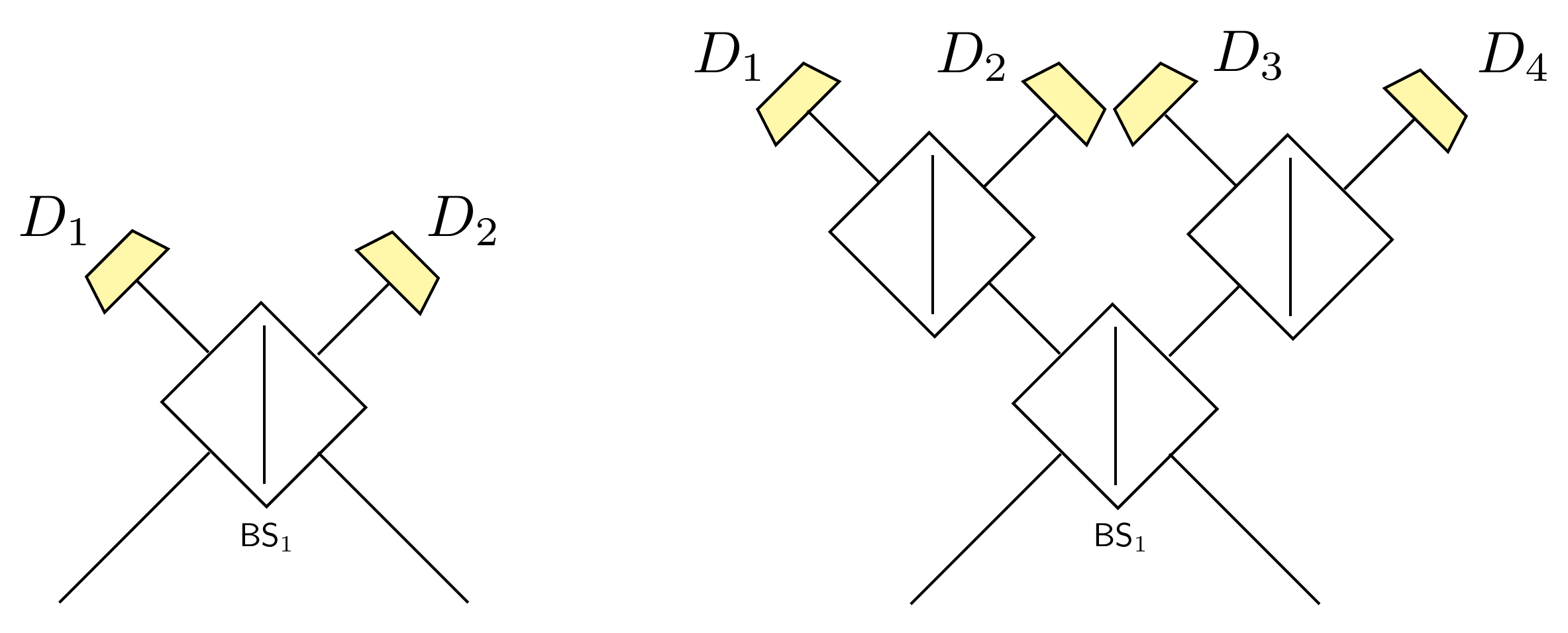}
		\caption{The detection setups \textbf{BS} (left) and \textbf{3BS} (right). The beam splitters are $(R,T)$ and non-polarizing. Unless otherwise specified, the detectors $D_i$ are conventional single photon click/no-click detectors.}
		\label{fig:setups}
	\end{figure}

	We define the following decision rules for the honest prover (for one detection window, corresponding to a round of the protocol):
	\begin{itemize}
		\item[(\textbf{BS})] Answer ``0'' if $D_1$ xor $D_2$ clicks, answer ``1'' if ($D_1$, $D_2$) click, answer ``$\varnothing$'' if no click occurs.
		\item[(\textbf{3BS})] Answer ``0'' if two clicks in one arm after $\mathsf{BS}_1$ are detected ($(D_1, D_2)$ or $(D_3, D_4)$), answer ``1'' if two clicks in different arms are detected ($(D_1, D_3)$, $(D_1, D_4)$, $(D_2, D_3)$ or $(D_2, D_4)$), else answer ``$\varnothing$''.
	\end{itemize}
	This means that in the \textbf{3BS} setup we post-select entirely on 2-click events, giving us weak NR, but only with some probability.
	
	In practice, no qubit or channel is perfect and we need to check under which conditions our protocol remains secure. To that end we will parametrise the entire setup from the single photon sources (at the verifiers) to the detection (at the prover) in terms of the errors that can appear. The setup consists of the following:
	
	\begin{itemize}
		\item Each verifier holds an imperfect single photon source, characterised by the probability that at least one photon is emitted $\eta_\text{source} = \Prob(n>0)$, the brightness $B = \Prob(n=1)$ and the accidental pair production rate $p_\text{pair}=\Prob(n=2)$, where $n$ is the number of single photons. We consider accidental multi-photon terms $\Prob(n>2)$ to be negligible.
		\item A communication channel between each verifier and the prover with a transmittance (at the prover) of $\eta_\text{BS}$\footnote{The beam splitter at \textsf{P} is where quantum interference between the incoming photons happens.}. We assume that both channels from $\mathsf{V_A}$ to \textsf{P} and from $\mathsf{V_B}$ to \textsf{P} have the same transmittance.
		\item The prover uses imperfect beam splitters with reflectance (amplitude) $R$ and transmittance (amplitude) $T$ as well as single photon detectors characterised by a detection efficiency $\eta_\text{det}$ (including loss between $\mathsf{BS}_1$ and the detectors as well as an imperfect intrinsic detection efficiency, per detector) and a dark count rate $p_\text{dark}$ (per detector).
		\item The final parameter is the overlap $\beta$ between the input states at the prover. Assuming that the equipment of both verifiers is identical, we can regard the photons leaving the sources as indistinguishable except in the degree of freedom we use to encode our quantum states in. One simple example would be the photon polarisation degree of freedom, such that $\beta = \lvert\braket{\psi|\phi}\rvert$ for polarisation qubits $\ket{\psi}$ and $\ket{\phi}$. In practice it may happen that a protocol round is started with a target overlap $\beta$ but the communication channel disturbs it to some $\tilde{\beta} = \beta + \delta$ with error $\lvert \delta \rvert > 0$.
	\end{itemize}

	\noindent We will denote the set of experimental parameters as 
	\begin{align}\label{equ:exp_tuple}
	    	\Omega_\beta = (\eta_\text{source}, B, p_\text{pair}, \eta_\text{BS}, \lvert R \rvert^2, \lvert T \rvert^2, \beta, \delta, \eta_\text{det}, p_\text{dark}).
	\end{align}
	Some comments about these parameters are to be made. From an experimental point of view the second order autocorrelation function (at zero time-delay) $g^{(2)}$, describing how bunched or anti-bunched the photons are coming from the source (see figure \ref{fig:bunching}), is easier to determine. The quantities $\eta_\text{source}$, $B$ and $g^{(2)}$ can be obtained in the lab and the latter has become a standard parameter to describe the quality of a single photon source \cite{trivedi2020generation}.
	
	
	\begin{figure}[h]
		\centering
		\includegraphics[width=0.5\textwidth]{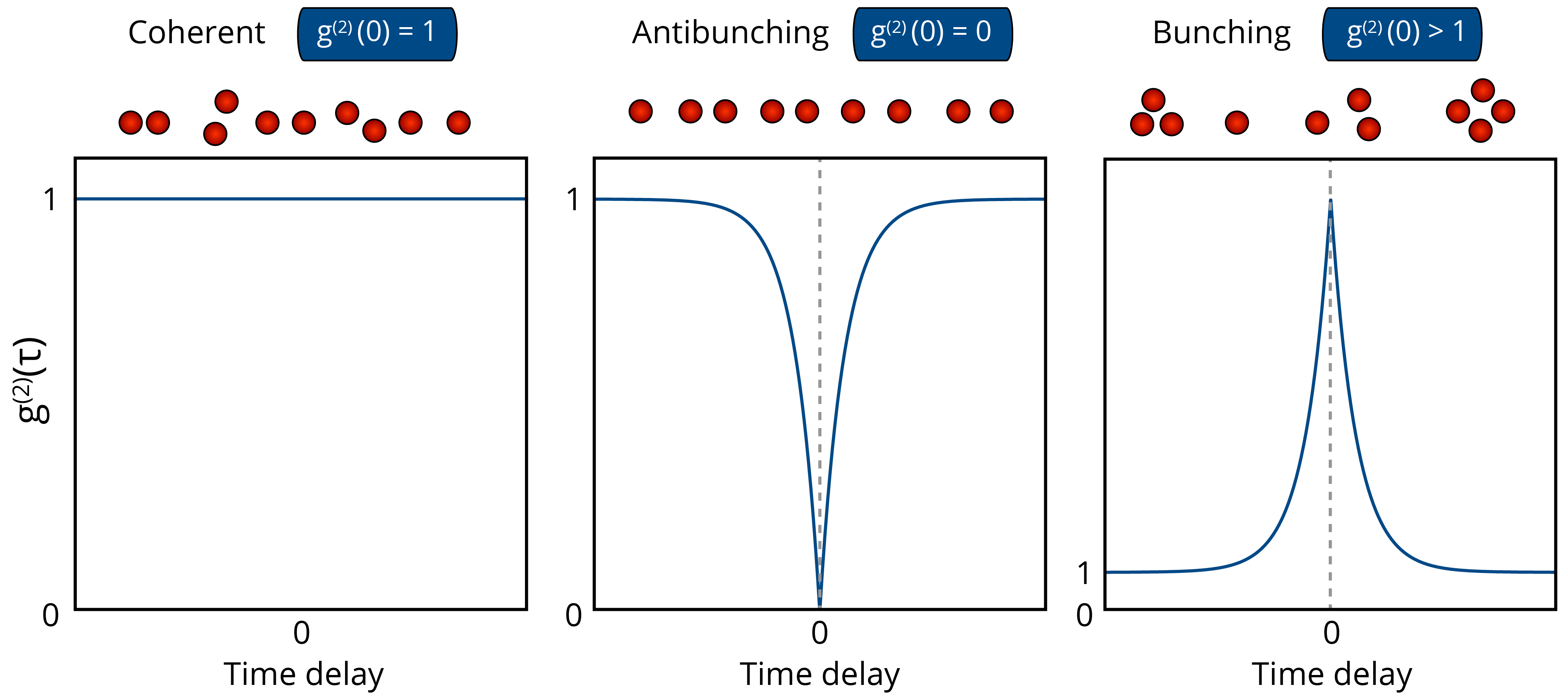}
		\caption{Bunching described by $g^{(2)}$ for photons coming from a single photon source. Taken from \texttt{https://blog.delmic.com/time-resolved-cathodoluminescence}.}
		\label{fig:bunching}
	\end{figure}
	
	 \noindent However, more important for us is the probability $p_\text{pair}$ that there accidentally \textit{are} two photons produced at one source, because this may influence the interference at the beam splitter with the photon from the other source. Assuming $\Prob(n > 2) \approx 0$, we can relate $g^{(2)}$ and $\Prob(n = 2)$ as follows, \cite{migdall2013single},
	\begin{align}
		p_\text{pair} = \Prob(n = 2) = \frac{g^{(2)}}{2} \mu^2 \approx \frac{g^{(2)}}{2} (2 \eta_\text{source} - B)^2,
	\end{align}
	with $\mu$ being the mean photon number produced by the source. Moreover, we account for errors in the communication channels as follows. If the verifiers expect a protocol with $\Omega_\beta$ and in reality the overlap between the input states changes from $\beta$ to $\tilde{\beta} = \beta + \delta$ on the way to the prover, then the honest party will run the protocol with $\Omega_{\tilde{\beta}} = \Omega_{\beta + \delta}$. That means \textsf{P} will reproduce the expected distribution a bit worse. How much error is tolerated for a given $\Omega_\beta$ will become evident in a later section. Note also, that the verifiers do not need to know/trust the detection system parameters but could in principle base their security checks simply on industry standard values of $\{ \eta_\text{det}, p_\text{dark} \}$. Moreover, $\eta_\text{BS}$ and $\delta$ can be estimated by the verifiers because they know in which way they send out the inputs (e.g. free space or in a quantum network). In this sense the verifiers have control over all parameters. Finally, we note that the protocol is based on two-photon interference and thus phase insensitive.
	
	\subsection{Imperfect honest prover}
	Here we will argue that all we are interested in are the two probability distributions $\Prob_{\Omega_\beta}(\mathsf{D}_1, \mathsf{D}_2)$ and $\Prob_{\Omega_\beta}(\mathsf{D}_1, \mathsf{D}_4)$ of these detector click patterns happening given an experimental configuration $\Omega_{\beta}$. This follows from the symmetry of the setup and the fact that photons bunch or anti-bunch into each bunch/anti-bunch output configuration with the same probability\footnote{Explicitly, we mean $\Prob((0,2)) = \Prob((2, 0))$, $\Prob((0,3)) = \Prob((3, 0))$ and $\Prob((1,2)) = \Prob((2, 1))$ for any $\lvert R \rvert^2$ and $\lvert T \rvert^2$.}, respectively, independent of the reflectance $\lvert R \rvert^2$ and transmittance $\lvert T \rvert^2$ of the beam splitter. Hence
	\begin{align}
			\begin{split}
					\Prob_{\Omega_\beta}(\mathsf{D}_1, \mathsf{D}_2) &= \Prob_{\Omega_\beta}(\mathsf{D}_3, \mathsf{D}_4), \\
					\Prob_{\Omega_\beta}(\mathsf{D}_1, \mathsf{D}_4) = \Prob_{\Omega_\beta}(\mathsf{D}_2, \mathsf{D}_4) &= \Prob_{\Omega_\beta}(\mathsf{D}_1, \mathsf{D}_3) = \Prob_{\Omega_\beta}(\mathsf{D}_2, \mathsf{D}_3).
			\end{split}
		\end{align}
	An intuitive example of this is the Hong-Ou-Mandel (HOM) effect. For overlap $\beta$ and an $(R,T)$ beam splitter, the probability to bunch to each output port is $\lvert R \rvert^2 \lvert T \rvert^2 (1+\beta^2)$ for both ports \cite{loudon2000quantum}. This uniform distribution on the output ports (given bunch or anti-bunch) also holds for the cases of 3 incoming photons as we prove later. That means
		\begin{align}
			\begin{split}
					\Prob_{\Omega_\beta}(0) &= \Prob_{\Omega_\beta}(\text{2-click in one arm}) = 2 \cdot \Prob_{\Omega_\beta}(\mathsf{D}_1, \mathsf{D}_2), \\
					\Prob_{\Omega_\beta}(1) &= \Prob_{\Omega_\beta}(\text{2-click in two arms}) = 4 \cdot \Prob_{\Omega_\beta}(\mathsf{D}_1, \mathsf{D}_4), \\
					\Prob_{\Omega_\beta}(\varnothing) &= 1 - \Prob_{\Omega_\beta}(0) - \Prob_{\Omega_\beta}(1).
			\end{split}
		\end{align}
	Finally, we post-select on conclusive answers and test the probability distributions of ``$0$'' and ``$1$'' answers there. So we are looking for
		\begin{align}\label{equ:exp_dist}
			\begin{split}
					p_{\Omega_{\beta}}(0) &\coloneqq \Prob_{\Omega_\beta}(0 \,|\, \text{concl.}) = \frac{\Prob_{\Omega_\beta}(0)}{1-\Prob_{\Omega_\beta}(\varnothing)}, \\
					p_{\Omega_{\beta}}(1) &\coloneqq \Prob_{\Omega_\beta}(1 \,|\, \text{concl.}) = \frac{\Prob_{\Omega_\beta}(1)}{1-\Prob_{\Omega_\beta}(\varnothing)}.
				\end{split}
		\end{align}
	Next, we find explicit expressions for the probability distributions of the needed detector click patterns. To that end, we expand
	\begin{equation}\label{equ: Pd1d2}
		\begin{alignedat}{3}
			&\Prob_{\Omega_\beta}(\mathsf{D}_1, \mathsf{D}_2) &&= \,\,&&(1-p_\text{dark})^2 \sum_{k} \Prob_{\Omega_\beta}(\mathsf{D}_1, \mathsf{D}_2 \,|\, k \text{ photons at } \mathsf{BS}_1) \Prob_{\Omega_\beta}(k \text{ photons at } \mathsf{BS}_1) \\
			&{}&&= \,\,&&(1-p_\text{dark})^2 \sum_k \bigg[ \Prob_{\Omega_\beta}(\mathsf{D}_1, \mathsf{D}_2 \,|\,\text{bunch}, k) \frac{\Prob_{\Omega_\beta}(\text{bunch} \,|\, k )}{2} \\
			&{}&&{}&&+ \Prob_{\Omega_\beta}(\mathsf{D}_1, \mathsf{D}_2 \,|\,\text{anti-bunch}, k)\Prob_{\Omega_\beta}(\text{anti-bunch} \,|\, k) \bigg] \Prob_{\Omega_\beta}(k \text{ photons at } \mathsf{BS}_1),
		\end{alignedat}
	\end{equation}
	where
		\begin{align}
			\Prob_{\Omega_\beta}(k \text{ photons at } \mathsf{BS}_1) = \sum_{\ell \geq k} \Prob_{\Omega_\beta}(k \text{ photons at } \mathsf{BS}_1 \,|\, \ell \text{ produced}) \Prob_{\Omega_\beta}(\ell \text{ produced}).
		\end{align}
	The factor $(1-p_\text{dark})^2$ accounts for the fact that the other two detectors $\mathsf{D}_3$ and $\mathsf{D}_4$ should not click. The factor $1/2$ in $\Prob_{\Omega_\beta}(\text{bunch} \, | \, k)/2$ stems from the fact that for a non-negligible contribution to $\Prob_{\Omega_\beta}(\mathsf{D}_1, \mathsf{D}_2)$ the photons have to bunch into the $(\mathsf{D}_1, \mathsf{D}_2)$ output arm, which happens with probability $1/2$. Otherwise, we would require both $\mathsf{D}_1$ and $\mathsf{D}_2$ to have a dark count (and loose all photons), which would dominate that contribution via $p_\text{dark}^2$\footnote{In practice values like $p_\text{dark}^2 \sim 10^{-14}$ per detection window can be achieved}. We neglect this other term. Furthermore, we consider up to $k=3$ photons incoming to $\mathsf{BS}_1$ and up to $\ell = 3$ photons in total being produced by the sources in a round of the protocol. The probabilities for $k\geq4$ incoming photons are considered negligible because assuming $\Prob(n>2) \approx 0$ for each source gives $\Prob_{\Omega_\beta}(4 \text{ photons at } \mathsf{BS}_1) \sim \eta_\text{BS}^4 p_\text{pair}^2$, which is only $\sim 10^{-12}$ for realistic conditions. Higher photon number terms are even smaller. In the process of treating the 3-photon case we had to generalise the HOM output port distribution to 3 photons, as formulated in the following lemma. The proof can be found in appendix \ref{app:3HOM_proof}.
	\begin{lemma}
		Consider photonic qubits $\ket{\psi}, \ket{\phi}$ arriving at one input port of a (symmetric) $(R,T)$ beam splitter and $\ket{\chi}$ at the other input port. Then the output port distribution is given by
		\begin{align}\label{HOM3}
			\Prob((3,0) \emph{ or } (0,3)) &= 4 \lvert R \rvert^2 \lvert T \rvert^2 \frac{\lvert \braket{\psi|\phi}\rvert^2 + \lvert \braket{\psi|\chi}\rvert^2 + \lvert\braket{\phi|\chi}\rvert^2}{2 \cdot (1 + \lvert\braket{\psi|\phi}\rvert^2)}, \\
			\Prob((2,1) \emph{ or } (1,2)) &= 1-\Prob((3,0) \emph{ or } (0,3)).
		\end{align} 
	\end{lemma}
	\noindent In the same vein we expand
\begin{equation}
	\begin{alignedat}{3}\label{equ: Pd1d4}
		&\Prob_{\Omega_\beta}(\mathsf{D}_1, \mathsf{D}_4) &&= \,\,&&(1-p_\text{dark})^2 \sum_{k} \Prob_{\Omega_\beta}(\mathsf{D}_1, \mathsf{D}_4 \,|\, k \text{ photons at } \mathsf{BS}_1) \Prob_{\Omega_\beta}(k \text{ photons at } \mathsf{BS}_1) \\
		&{}&&= \,\,&&(1-p_\text{dark})^2 \sum_k \bigg[ \Prob_{\Omega_\beta}(\mathsf{D}_1, \mathsf{D}_4 \,|\,\text{bunch}, k) \Prob_{\Omega_\beta}(\text{bunch} \,|\, k ) \\
		&{}&&{}&&+ \Prob_{\Omega_\beta}(\mathsf{D}_1, \mathsf{D}_4 \,|\,\text{anti-bunch}, k)\Prob_{\Omega_\beta}(\text{anti-bunch} \,|\, k) \bigg] \Prob_{\Omega_\beta}(k \text{ photons at } \mathsf{BS}_1).
	\end{alignedat}
\end{equation}
Here there is no factor $1/2$ after $\Prob_{\Omega_\beta}(\text{bunch} \,|\, k )$ because no matter into which output arm the photons bunch only one extra detector dark count is needed, incurring only a factor $p_\text{dark}$ instead of $p_\text{dark}^2$, which we do not neglect. \\
Given all this, one has to explicitly write out all these probability contributions in \eqref{equ: Pd1d2} and \eqref{equ: Pd1d4}, which gives very long overall expressions. These can be found in appendix \ref{app:Pd1d2_Pd1d4}. Finally, note that for example for overlaps $\beta \in \{0,1\}$ an ideal honest party with the set of parameters $\Omega_{\beta} = (1, 1, 0, 1, 1/2, 1/2, \beta, 0, 1, 0)$ will produce
	\begin{alignat}{2}
			\Prob(0 \,|\, \rho_0, \text{concl.}) &= \frac{1}{3} \qquad \Prob(0 \,|\, \rho_1, \text{concl.}) &&= 1, \\
			\Prob(1 \,|\, \rho_0, \text{concl.}) &= \frac{2}{3} \qquad \Prob(1 \,|\, \rho_1, \text{concl.}) &&= 0, \\
			\Prob(\varnothing \,|\, \rho_0) &= \frac{1}{4} \qquad \Prob(\varnothing \,|\, \rho_1) &&= \frac{1}{2}
	\end{alignat}
using the \textbf{3BS} setup. This can be generalised to any $\beta$ as 
	\begin{align}\label{equ:pinc_3bs}
		\begin{split}
			\Prob(0 \,|\, \rho_\beta, \text{concl.}) &= \frac{1 + \beta^2}{3 - \beta^2}, \\
			\Prob(1 \,|\, \rho_\beta, \text{concl.}) &= 1- \frac{1 + \beta^2}{3 - \beta^2}, \\
			\Prob(\varnothing \,|\, \rho_\beta) &= \frac{1 + \beta^2}{4}.
		\end{split}
	\end{align}
Since an ideal \textsf{P} would produce this with the proposed setup it makes sense to consider this one as ``the ideal distribution'' for an experiment instead of the usual SWAP-test distribution. Nevertheless, the essential quantum interference happens in the first beam splitter $\mathsf{BS}_1$ implementing the SWAP-test.
	
\subsection{Statistical testing}
The SWAP-test is a probabilistic measurement and in a realistic scenario with errors it also won't give a deterministic answer on identical inputs of the form $\ket{\psi}\otimes \ket{\psi}$. In order to distinguish the honest prover from attackers the verifiers therefore need to test between the hypotheses ``the sample received comes from \textsf{P}'' and ``the sample received comes from attackers'', considering that \textsf{P} will make some (predictable) errors. As each round is run independently the samples generated over many rounds will be samples from a binomial distribution. In the problem there are three involved distributions: the ideal one, the imperfect honest one, and the attacker one. To distinguish them, we will perform a binomial test. Similar to the ideal scenario, we will define acceptance regions around the ideal distributions for each $\beta$ in such a way that we still capture the imperfect \textsf{P} with high probability. These acceptance regions depend on the experimental conditions $\Omega_{\beta}$. If the conditions are too bad, these regions are forced to be wide, possibly even largely overlapping for different $\beta$. Then the test will be impaired and, for example, if all acceptance regions overlap, the protocol can be broken. Or else, the conditions could be bad enough so that we have to run an infeasibly large amount of rounds. In such cases we say \textit{the experimental conditions are too weak for QPV}. Intuitively, if the experimental conditions are good enough, the more rounds we run, the narrower the acceptance regions will become\footnote{While still accepting \textsf{P} with high probability, because we will define the acceptance region accordingly, see further below in the main text.} and the lower the probability that attackers produce a sample which reaches all acceptance regions simultaneously. This behaviour is depicted in figure~\ref{fig:stattest_squeezeR}. Then we can see how many rounds we need to run in order to achieve enough confidence in distinguishing attackers from \textsf{P} as a function of the experimental parameters $\Omega_{\beta}$. The worse the conditions, the more rounds we will need to run and if the conditions are too weak for QPV the protocol can be broken.
	
	\begin{figure}[h]
		\centering
		\includegraphics[width=0.66\textwidth]{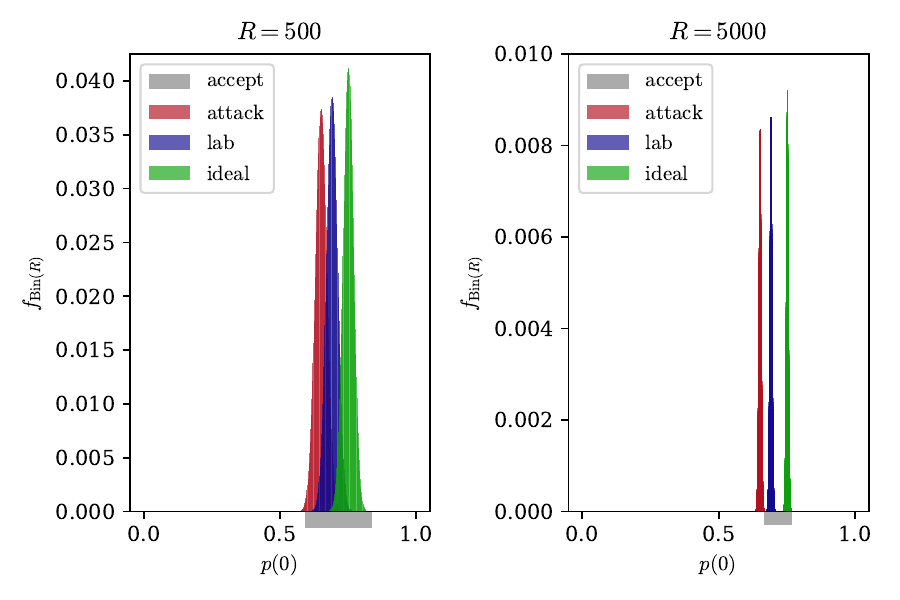}
		\caption{Sketch of ``squeezing'' the attackers out of an acceptance region. On the left, we have not run enough rounds yet and a large part of the attacker distribution (red) overlaps with the acceptance region (gray). Thus attackers would have a decently high probability of returning a sample that gets accepted. On the right, after many more rounds, the probability of returning a sample that lies in the acceptance region is negligibly small. These plots are for an example overlap of $\beta = 3/4$, $f_\text{Bin}(R)$ denotes the binomial probability density function for $R$ rounds and $p(0)$ the fraction of ``0'' results.}
		\label{fig:stattest_squeezeR}
	\end{figure}

\noindent In order to still capture the honest prover with high probability even in the presence of errors, we need to widen the acceptance regions of the errorless protocol as given by \eqref{equ:accreg_ideal}. The new lower bound will be the smaller of the two $\alpha$-quantiles of the ideal and the imperfect distribution. The new upper bound will be the larger of the two $(1-\alpha)$-quantiles. In other words,
\begin{align}\label{equ:accreg_exp}
	\begin{split}
		L_{\alpha, \Omega_{\beta}} &= \min \left\{ z_{\alpha}(\beta, R_\beta), z_{\alpha}(\Omega_{\beta}, R_\beta) \right\} / R_\beta \\
		U_{\alpha, \Omega_{\beta}} &= \max \left\{ z_{1-\alpha}(\beta, R_\beta), z_{1-\alpha}(\Omega_{\beta}, R_\beta) \right\} / R_\beta .
	\end{split}
\end{align}
Here again the values of $z_{q}(\beta,R_\beta)$ and $z_{q}(\Omega_\beta,R_\beta)$ can be obtained via the inverse of the cumulative distribution function $F_{\text{Bin}(R_\beta, p_\beta(0))}$ and $F_{\text{Bin}(R_\beta, p_{\Omega_\beta}(0))}$, respectively. This defines the round-dependent acceptance regions $\mathsf{acc}_{\Omega_\beta}(\alpha, R_\beta) = [L_{\alpha, \Omega_{\beta}}, U_{\alpha, \Omega_{\beta}}]$. Defining these region in this way ensures that we still capture \textsf{P} with high probability $\geq 1 - O(k\alpha)$, with $k$ the number of different overlaps used in the protocol and $\alpha$ can be set very small, like $10^{-6}$. Meanwhile, attackers need to get $\hat{p}_\beta(0) \in \mathsf{acc}_{\Omega_\beta}(\alpha, R_\beta)$ for all $\beta$ in order to succeed. If the experimental conditions $\Omega_\beta$ are so bad that all $\mathsf{acc}_{\Omega_\beta}(\alpha, R_\beta)$ overlap, attackers can succeed by choosing to answer with some fixed list producing $\hat{p}(0) \in \bigcap_\beta \mathsf{acc}_{\Omega_\beta}(\alpha, R_\beta)$. Indeed, all acceptance regions overlap when one tries to implement the SWAP-test with just one beam splitter and two click/no-click detectors. We hence \textit{demand} that not all acceptance regions overlap, that is $\bigcap_\beta \mathsf{acc}_{\Omega_\beta}(\alpha, R_\beta) = \emptyset$. \\
Finally, one subtlety is that with the proposed $\textbf{3BS}$ setup the rate of inconclusive answers is overlap dependent, cf. equation \eqref{equ:pinc_3bs}. We still want to keep a uniform distribution over the input overlaps $\{ \beta_1, \dots, \beta_k \}$, though. To do so, we send uniformly random input states $\rho_\beta$ until the number of conclusive rounds reaches $R_\beta \geq R_\text{threshold}$ for all $\beta$.  Eventually, with sufficiently nice $\Omega_\beta$ and sufficiently high $R_\text{threshold}$, we may achieve
\begin{align}\label{equ:p_acc_att}
	\Prob_{\Omega_{\beta}}(\mathsf{acc} | \mathsf{att}) = \prod_\beta \Prob_{\Omega_{\beta}}(\mathsf{acc}_\beta | \mathsf{att}) = \prod_\beta \Prob( \hat{p}_0(\beta) \in \mathsf{acc}_{\Omega_\beta}(\alpha, R_\beta) ) \leq \eps,
\end{align}
with any desired $\eps$. For example, we could set $\eps = \alpha$ and, in the end, choose $\alpha$ very small, say $\alpha \sim 10^{-6}$. Then we would accept \textsf{P} with high probability at least $1-O(k\alpha)$ and accept attackers with vanishing probability at most $\alpha$.

	\begin{figure}[h]
		\centering
		\includegraphics[width=0.66\textwidth]{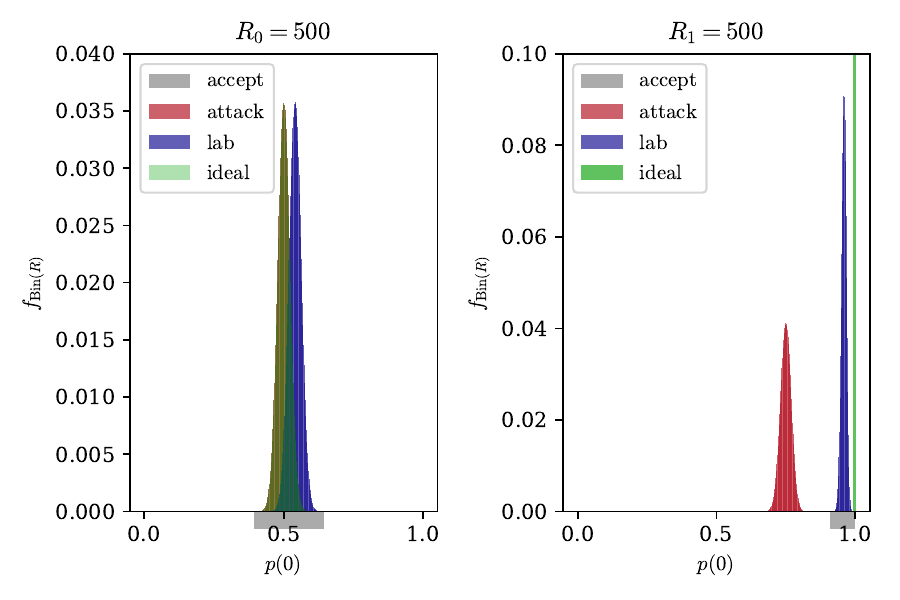}
		\caption{Illustrating the test via the example of the protocol with overlaps $\{ 0,1 \}$, so that $p_{\beta=0}(0) = 1/2$ and $p_{\beta=1}(0) = 1$. We define acceptance regions as in \eqref{equ:accreg_exp} and the difficulty for the attackers is that they have to get into the acceptance regions for all overlaps $\beta$ to pass the test. If the experimental conditions are good enough, so that the acceptance regions can be made sufficiently narrow by increasing the number of rounds $R_\beta$, the probability that attackers return a successful sample is negligibly small, because their probability density function has essentially zero overlap with at least one acceptance region. Note that in the left plot the attack distribution is the same as the ideal one (hence overlaying red and green), but they perform badly on identical inputs.}
		\label{fig:stattest_overlaps01}
	\end{figure}
	
\subsection{Optimal attacker strategy}
	We now go on to describe optimal attack strategies restricted to classical communication and no pre-shared entanglement in the presence of experimental errors. Our result on parallel repetition implies that the optimal strategy for multiple rounds is to use the optimal single-round strategy many times. Attackers will therefore maximise their chance to get accepted by trying to bring their $\hat{p}_\beta(0)$ as close as possible to $\max \{ p_\beta(0), p_{\Omega_{\beta}}(0) \}$ for all $\beta$. This is because generally attackers want to answer ``0'' as often as they can\footnote{while performing as well as possible overall} in order to perform better on high overlap inputs. We take this into account by allowing them to adaptively optimise towards $\max \{ p_\beta(0), p_{\Omega_{\beta}}(0) \}$. The relevant parameter is therefore\footnote{Since, conditioned on a conclusive answer, we have $\Delta_\beta(0) = \Delta_\beta(1)$ we can just use the ``0'' answers and write $\Delta_\beta$.} $\Delta_\beta = \lvert \max \{ p_\beta(0), p_{\Omega_{\beta}}(0) \} - \hat{p}_\beta(0) \rvert$. Because attackers have to minimize $\Delta_\beta$ for all $\beta$ simultaneously to have a chance to win, we choose to minimize the one-norm $\lVert \Delta \rVert_1$ of the vector
	\begin{align}
		\Delta = \begin{pmatrix}
			\Delta_{\beta_1} \\ \vdots \\ \Delta_{\beta_k}
		\end{pmatrix}.
	\end{align}
	As before the attackers are restricted to PPT-measurements $\{ \Pi_0, \Pi_1, \Pi_{\varnothing} \}$ to capture attacks using classical communication. In section~\ref{sec:qpvswaploss} we proved full loss tolerance of our protocol. Adding two more beam splitters in the \textbf{3BS} setup does not change that, as the quantum interference that is hard to simulate for attackers happens in $\mathsf{BS}_1$ and the effect of the extra beam splitters can be classically calculated by each attacker. We can hence focus on the case of $\eta = 1$ from now on. In the \textbf{3BS} setup there is a non-zero inconclusive-answer rate even in the $\eta = 1$ case and it is overlap dependent, which adds extra difficulty compared to the original setting. We assume they have some way of getting the honest loss pattern \eqref{equ:pinc_3bs} right by adding the constraints\footnote{If they don't get it right, they are caught right away.} 
	\begin{align}\label{equ:lossconstraints}
		\Tr[\Pi_\varnothing \rho_\beta] = \frac{1 + \beta^2}{4} \qquad \forall \beta.
	\end{align}
	In total, this leaves us with the following optimisation problem:
	\begin{equation}
		\label{equ:Opt}
		\begin{alignedat}{2}
			&\texttt{minimize:} \qquad && \lVert \Delta \rVert_1 \\
			&\texttt{subject to:} \qquad &&\Pi_{0}, \Pi_{1}, \Pi_{\varnothing} \succeq 0 \\
			&{} &&\Pi_{0}^{T_\mathsf{B}}, \Pi_{1}^{T_\mathsf{B}}, \Pi_{\varnothing}^{T_\mathsf{B}} \succeq 0 \\
			&{} &&\Pi_0 + \Pi_{1} + \Pi_{\varnothing} = \mathds{1} \\
			&{} &&\Tr[\Pi_\varnothing \rho_\beta] = \left(1 + \beta^2\right)/4 \qquad \forall \beta,
		\end{alignedat}
	\end{equation}
	with
	\begin{align}
		\Delta_i = \begin{cases}
			p_{\beta_i}(0) - \frac{\Tr[\Pi_0 \rho_{\beta_i}]}{1-\left(1 + \beta_i^2\right)/4} & \text{if } p_{\beta_i}(0) \geq p_{\Omega_{\beta_i}}(0) \\
			p_{\Omega_{\beta_i}}(0) - \frac{\Tr[\Pi_0 \rho_{\beta_i}]}{1-\left(1 + \beta_i^2\right)/4} & \text{if } p_{\beta_i}(0) \leq p_{\Omega_{\beta_i}}(0)
		\end{cases}.
	\end{align}
	The solution will give us the optimal PPT-measurement attackers can apply to do as well as possible on all $\beta$ for the statistical test which determines the final success probability in the end.\footnote{Note that minimizing $\Delta = p-\hat{p}$ also minimizes the Kullback-Leibler divergence $D_\text{KL}(P \parallel Q)$ for $P \sim \text{Bin}(n,p)$ and $Q \sim \text{Bin}(n,\hat{p})$.} Since practically one will have $p_\text{pair} > 0$, also attackers will have access to three or four photons sometimes and possibly they can do better with these extra resources. Therefore, we will also solve the above optimization problem for higher photon numbers by adjusting the dimensions of the involved operators. In particular, for $k$ photons attackers will apply a POVM $\left\{ \Pi_0^{(k)}, \Pi_1^{(k)}, \Pi_{\varnothing}^{(k)} \right\}$ on $2^k$ dimensional states $\rho^{(k)}_\beta$. In general the state prepared by the verifiers takes the form
	\begin{align}\label{equ:Wernerk}
		\rho_\beta^{(k)} = \int_{\text{U}(2)} U^{\otimes k} P_{\psi \phi}^{(k)} \left( U^\dag \right)^{\otimes k} d\mu(U),
	\end{align}
	where $P_{\psi \phi}^{(k)}$ is some pure $k$-qubit state describing two states $\ket{\psi}, \ket{\phi}$ with $\lvert \braket{\psi|\phi} \rvert = \beta$ making up a $k$-photon state\footnote{For example, if one source produces a pair we'd have \smash{$P_{\psi \phi}^{(3)} = \ketbra{\psi \psi \phi}{\psi \psi \phi}$}}. Here $\mu$ is the Haar measure on the unitary group U$(2)$. Integrals of the form \eqref{equ:Wernerk} can be explicitly calculated using Weingarten calculus \cite{weingarten1978asymptotic,collins2006integration}. We used the Mathematica package \textsf{IntU} \cite{puchala2011symbolic} to calculate $\rho_\beta^{(3)}$ and $\rho_\beta^{(4)}$.

	The above optimization then gives $\Delta_\text{min}^{(k)}$. Clearly, it is beneficial for attackers to choose to answer as much as possible in rounds with more photons. The overall $\Delta_{\beta, \text{min}}$ will then be composed as 
	\begin{align}
		\Delta_{\beta, \text{min}} = p_{(2)} \Delta_{\beta, \text{min}}^{(2)} + p_{(3)}\Delta_{\beta, \text{min}}^{(3)} + p_{(4)}\Delta_{\beta, \text{min}}^{(4)},
	\end{align}
	where $p_{(m)}$ is the fraction of rounds of $m$ photons among the attackers' conclusive answer rounds. In particular, the attackers control $p_{(m)}$ and it could be that $p_{(4)}=1$ and $p_{(2)}=p_{(3)}=0$ for example. This is constrained by \eqref{equ:lossconstraints}. Say, for example, the verifiers expect a conclusive-answer rate of $10^{-9}$ and the rate of 4 photons (double pair-production) is $p_4 \sim 10^{-8}$. Then indeed attackers can choose $p_{(4)}=1$ and only answer on 4 photon rounds (and even among those not answer on all). If $p_\text{concl.} \sim 10^{-6}$ in this example, then attackers will also need to answer on some rounds with 3 photons, thus $p_{(4)} < 1$ and $p_{(3)} > 0$, and possibly also $p_{(2)} > 0$ in order to be able to answer conclusively often enough. All this will affect the $\Delta_{\beta, \text{min}}$ and in turn the total $\Delta_{\text{min}}$, which will then affect the statistical test in the end, which affects how we have to set $R_\text{threshold}$.
	
	\section{Simulation under realistic conditions}\label{sec:simulationres}
	We have done all these simulations for the prime example of overlaps $\{ 0, 1 \}$, that is, sending either equal or orthogonal states. First of all, in that case the optimisation \eqref{equ:Opt} gives 
	\begin{align}
		\begin{split}
			\lVert \Delta_{\text{min}}^{(2)} \rVert_1 &= \frac{1}{2}, \\
			\lVert \Delta_{\text{min}}^{(3)} \rVert_1 &= \frac{1}{3}, \\
			\lVert \Delta_{\text{min}}^{(4)} \rVert_1 &= \frac{1}{6}.
		\end{split}
	\end{align}
	The portions $p_{(m)}$ are chosen adaptively, depending on what the overall expected experimental conclusive-rate $p_\text{concl.} = 1 - \frac{1}{k} \sum_\beta p_\varnothing(\Omega_\beta)$ is\footnote{Basically, the attackers use higher $m$ as often as possible before going on to use $m-1$, as mentioned above.}. Then everything is fed into codes calculating all the $\mathsf{acc}_{\Omega_\beta}(\alpha, R_\beta)$, which depend on the experimental conditions, the number of conclusive rounds we run and how small we set $\alpha$. Finally, we increase the number of (conclusive) rounds and plot $\Prob_{\Omega_{\beta}}(\mathsf{acc} \,|\, \mathsf{att})$ as a function of $R_\text{threshold}$. Realistic experimental parameters could be close to \cite{trivedi2020generation, migdall2013single}
	\begin{equation}\label{equ:exp_params}
		\begin{alignedat}{3}
				\eta_\text{source} &= 0.12 \hspace{5.2cm}\eta_\text{BS} &&= 0.20  \hspace{2.1cm} \eta_\text{det} &&= 0.20\\
				B &= 0.1197\dots  \hspace{4.2cm}\lvert R \rvert^2 &&= 0.45 \hspace{1.5cm} p_\text{dark} &&= 10^{-7}\\
				g^{(2)} &= 0.04 \hspace{5.2cm}\lvert T \rvert^2 &&= 0.55 &{}\\
				p_\text{pair} &= \frac{g^{(2)}}{2} (2 \eta_\text{source} - B)^2 \approx 3 \cdot 10^{-4} &{} &{}
		\end{alignedat}
	\end{equation}
	We wrote code that calculates $\Prob_{\Omega_{\beta}}(\mathsf{acc} \,|\, \mathsf{att})$ as a function of $R_\text{threshold}$. The results for both the original \textbf{BS} setup and the proposed \textbf{3BS} setup are depicted in figures~\ref{fig:p_acc_att_diff_polerr_swap} and ~\ref{fig:p_acc_att_diff_polerr}, respectively.

	\begin{figure}
		\centering
		\begin{subfigure}{.375\textwidth}
			\centering
			\captionsetup{width=.3\linewidth}
			\includegraphics[width=\linewidth]{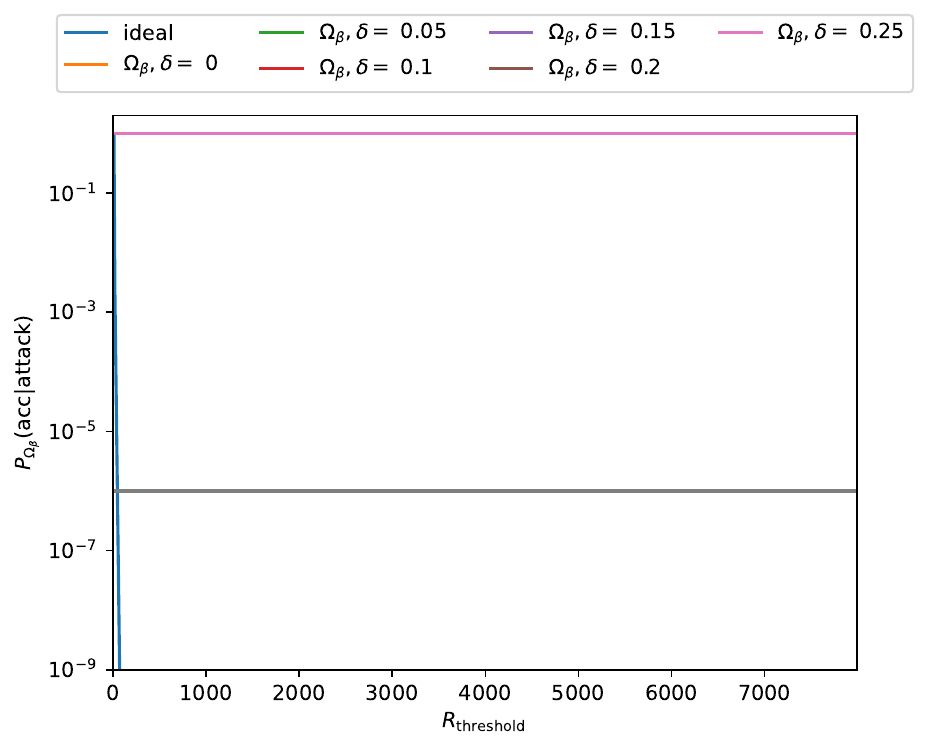}
			\caption{\textbf{BS.}}
			\label{fig:p_acc_att_diff_polerr_swap}
		\end{subfigure}%
		\begin{subfigure}{.375\textwidth}
			\centering
			\captionsetup{width=.3\linewidth}
			\includegraphics[width=\linewidth]{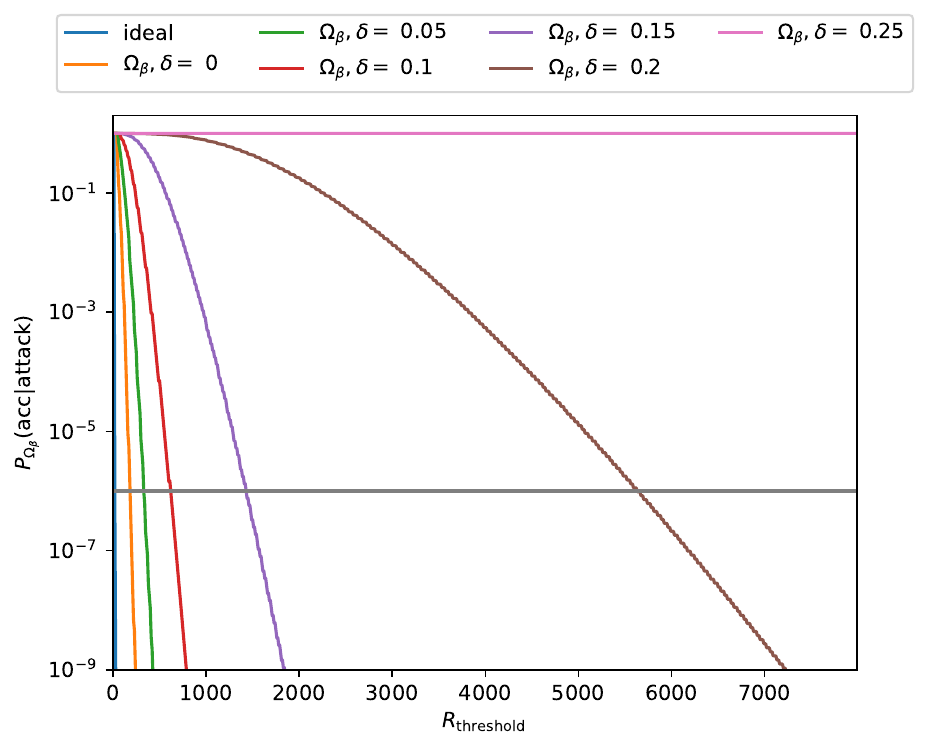}
			\caption{\textbf{3BS.}}
			\label{fig:p_acc_att_diff_polerr}
		\end{subfigure}
	\caption{The success probability of attackers, $\Prob_{\Omega_{\beta}}(\mathsf{acc} | \mathsf{att})$ given the experimental conditions ~\eqref{equ:exp_params} and different values of the overlap error $\delta$ as a function of $R_\text{threshold}$ for the different setups. Here $\beta \in \{ 0, 1 \}$. The grey horizontal line marks $\Prob_{\Omega_{\beta}}(\mathsf{acc} | \mathsf{att}) = 10^{-6}$. As mentioned in the main text, for the original \textbf{BS} setup without NR detectors high loss (from source to detection) forces $p_{\Omega_{\beta}}(0) \approx 1$ for all $\beta$ such that all acceptance regions will overlap, making the protocol insecure (note that all imperfect lines lie on top of each other). The \textbf{3BS} setup removes this insecurity by introducing probabilistic NR using just simple click/no-click detectors.}
	\end{figure}

	Note that the $R_\text{threshold}$ in these plots correspond to \textit{conclusive rounds per overlap} $\beta \in \{ \beta_1, \dots, \beta_k \}$. If, for example, $p_\text{concl.} \sim 10^{-6}$ and $R_\text{threshold} \sim 10^3$, then we'll need to run $\sim k \cdot 10^9$ rounds in total in order to build up enough conclusive rounds (per $\beta$). In general, if $p_\text{concl.} \sim 10^{-a}$ and $R_\text{threshold} \sim 10^b$, then we expect to run at least $\sim k \cdot 10^{a+b}$ rounds in total. So if in principle security can be achieved, meaning that we can make $\Prob_{\Omega_{\beta}}(\mathsf{acc} | \mathsf{att})$ arbitrarily small by increasing the number of rounds we run, the time needed per round will then determine if QPV is practically feasible under the experimental conditions $\Omega_{\beta}$.

\section{Conclusion}
We constructed and analysed a new quantum position verification protocol, QPV$_{\textsf{SWAP}}$, and showed that it possesses several desirable properties. It was shown that it is fully loss tolerant against LOCC attackers with no pre-shared entanglement, that it can be attacked with $\tilde{O}(n)$ pre-shared EPR pairs and that at least $\sim 0.103 n$ pre-shared EPR pairs are necessary in the $\beta \in \{ 0, 1\}$ case. Moreover, it fulfills strong parallel repetition and retains the loss tolerance even if all rounds are run in parallel. QPV$_{\textsf{SWAP}}$ even remains loss tolerant and secure if attackers are allowed to use quantum communication between them (without pre-shared entanglement), making it the first QPV protocol with this property. However, we were not able to show a finite gap in the loss+quantum communication scenario between the attacker and the honest success probability. We suspect that there is a finite gap and that this also holds for the $n$ round parallel repetition, but leave the proofs to future work.

In addition, the flexibility and simplicity of the SWAP test, both theoretically and experimentally, make it an excellent candidate for practical QPV. To that end, we undertook a detailed analysis of our protocol under realistic experimental conditions, in which we quantify the entire experimental setup in terms of possible imperfections and take these into account in the attack model. In the end a binomial test determines (with high probability) if the returned list of answers originated from \textsf{P} or from attackers. We identified a condition indicating if the experimental conditions $\Omega_\beta$ allow for security in principle and if so, calculated the ``figure of merit'' $R_\text{threshold}(\Omega_\beta)$, which is the number of conclusive rounds (per $\beta$) we need to collect in order to be able to guarantee a sufficiently low attack success probability $\Prob_{\Omega_\beta}(\mathsf{acc} | \mathsf{attack})$. For the prime example of sending either identical or orthogonal states and realistic conditions already $R_\text{threshold}(\Omega_\beta) \sim 10^2$-$10^3$ suffices to achieve $\Prob_{\Omega_\beta}(\mathsf{acc} | \mathsf{attack}) \leq 10^{-6}$. Our protocol therefore remains fairly robust in the presence of experimental imperfections and the challenge for implementation is to run many rounds fast enough in order to be able to collect $R_\text{threshold}(\Omega_\beta)$ conclusive rounds for each overlap.

\paragraph{Acknowledgments}
We would like to thank Wolfgang L\"offler, Kirsten Kanneworff and Norbert L\"utkenhaus for many useful discussions. RA and HB were supported by the Dutch Research Council (NWO/OCW), as part of the Quantum Software Consortium programme (project number 024.003.037). PVL and HB were supported by the Dutch Research Council (NWO/OCW), as part of the NWO Gravitation Programme Networks (project number 024.002.003).

\pagebreak

\bibliographystyle{alpha}
\bibliography{ZoteroReferences,QPV_folder,HOM_folder}

\begin{appendices}
\section{\texorpdfstring{QPV$_{\textsf{SWAP}}$}{QPVswap}}
\subsection{Exponential suppression of attacker success probability}
\label{AppendixExpSuppr}
Here we provide the proof of \eqref{equ:psuccattgen}. Let $N_\beta$ the binomial distributed random variable  describing the number of ``0'' answers of attackers in $L_\beta$. Since $p_\beta, p^{\mathsf{AB}}_\beta \in \big[ \frac{1}{2}, 1 \big)$, we may approximate the binomial distribution with a normal distribution $\mathcal{N}(\mu, \sigma^2)$ with $\mu = R_\beta p$ and $\sigma^2 = R_\beta p (1-p)$ for $p \in \{ p_\beta, p^{\mathsf{AB}}_\beta \}$, respectively. This is valid as long as $R_\beta (1-p)$ is sufficiently large, which we can always achieve by making $R_\beta$ big enough. Then
\begin{align*}
    \Prob(\mathsf{acc}_\beta | \mathsf{attack}) &= \Prob\left(z_{\frac{\alpha}{2}}^\mathsf{P} \leq N_\beta \leq z_{1-\frac{\alpha}{2}}^\mathsf{P} \right) \\ &= F_{N_\beta}\left(z_{1-\frac{\alpha}{2}}^\mathsf{P}\right) - F_{N_\beta}\left(z_{\frac{\alpha}{2}}^\mathsf{P}-1\right)\\
    &\approx \frac{1}{2}\left[ 1+ \operatorname{erf}\left( \frac{z_{1-\frac{\alpha}{2}}^\mathsf{P} - \mu_\beta^{\mathsf{AB}}}{\sqrt{2}\sigma_\beta^{\mathsf{AB}}} \right) \right] - \frac{1}{2}\left[ 1+ \operatorname{erf}\left( \frac{z_{\frac{\alpha}{2}}^\mathsf{P} - \mu_\beta^{\mathsf{AB}}}{\sqrt{2}\sigma_\beta^{\mathsf{AB}}} \right) \right].
\end{align*}
Now for $\mathcal{N}(\mu, \sigma^2)$ one has $z_q = F^{-1}(q) = \mu + \sqrt{2} \sigma \operatorname{erf}^{-1}(2q-1)$. Replacing the $z_q$ values and defining $c_\alpha \coloneqq \operatorname{erf}^{-1}(1-\alpha)$ as well as $f_{\beta}^\mathsf{X} = \sqrt{2 p_\beta^\mathsf{X} (1-p_\beta^\mathsf{X})}$ gives
\begin{align*}
    \Prob(\mathsf{acc}_\beta | \mathsf{attack}) \approx \frac{1}{2} \operatorname{erf}\left( \frac{\sqrt{R_\beta} \Delta_\beta + f_{\beta}^\mathsf{P} c_\alpha}{f_{\beta}^\mathsf{AB}} \right) - \frac{1}{2} \operatorname{erf}\left( \frac{\sqrt{R_\beta} \Delta_\beta - f_{\beta}^\mathsf{P} c_\alpha}{f_{\beta}^\mathsf{AB}} \right).
\end{align*}
Using that $\operatorname{erf}(x) \approx 1 - \frac{e^{-x^2}}{\sqrt{\pi}x}$ for large $x$, we can write
\begin{align*}
    \Prob(\mathsf{acc}_\beta | \mathsf{attack}) \approx \sqrt{\frac{2}{\pi}}f_{\beta}^\mathsf{AB} \left[ \frac{e^{-\left(\sqrt{R_\beta}\Delta_\beta - f_{\beta}^\mathsf{P} c_\alpha \right)^2 / \left(f_{\beta}^\mathsf{AB}\right)^2}}{\sqrt{R_\beta}\Delta_\beta - f_{\beta}^\mathsf{AB} c_\alpha} - \frac{e^{-\left(\sqrt{R_\beta}\Delta_\beta + f_{\beta}^\mathsf{P} c_\alpha\right)^2/ \left(f_{\beta}^\mathsf{AB}\right)^2}}{\sqrt{R_\beta}\Delta_\beta + f_{\beta}^\mathsf{AB} c_\alpha} \right]
\end{align*}
As $p_\beta^\mathsf{AB} \neq 0$ and $p_\beta^\mathsf{AB} \neq 1$, we may neglect the terms $f_{\beta}^\mathsf{AB} c_\alpha$ in the denominators because we can make $R_\beta$ sufficiently large. Moreover, leaving out the second (positive) exponential term gives the approximate upper bound
\begin{align*}
    \Prob(\mathsf{acc}_\beta | \mathsf{attack}) \lesssim \frac{\sqrt{2}f_{\beta}^\mathsf{AB}}{\sqrt{\pi R_\beta} \Delta_\beta} e^{-\left(\sqrt{R_\beta}\Delta_\beta - f_{\beta}^\mathsf{P} c_\alpha \right)^2 / \left(f_{\beta}^\mathsf{AB}\right)^2} = O \left( \frac{2^{-\Delta_\beta^2 R_\beta}}{\Delta_\beta \sqrt{R_\beta}}\right).
\end{align*}

\subsection{Relating \texorpdfstring{$p_\text{succ}$ to $\lVert \Delta \rVert_1$ for QPV$_{\textsf{SWAP}}(0,1)$}{psucc to |Delta| for QPVswap(0,1)}}\label{AppendixPsuccDelta}
Relating these two quantities is fairly straigtforward and achieved by one application of the triangle inequality. Consider
\begin{align*}
    p_\text{succ} = \frac{1}{2}\frac{\Tr[\Pi_0 \rho_0]}{\eta} + \frac{1}{2}\frac{\Tr[\Pi_1 \rho_1]}{\eta} \leq u,
\end{align*}
with $u \leq 3/4$. We want to massage this in order to get $\Delta_0$ and $\Delta_1$ expressions into it. Doing so gives
\begin{align*}
     1-\frac{\Tr[\Pi_0 \rho_0]}{\eta} + \frac{1}{2}-\frac{\Tr[\Pi_1 \rho_1]}{\eta} \geq \frac{3}{2} - 2u.
\end{align*}
This implies
\begin{align*}
    \lVert \Delta \rVert_1 = \bigg| 1-\frac{\Tr[\Pi_0 \rho_0]}{\eta} \bigg| + \bigg| \frac{1}{2}-\frac{\Tr[\Pi_1 \rho_1]}{\eta} \bigg| \geq \bigg| 1-\frac{\Tr[\Pi_0 \rho_0]}{\eta} + \frac{1}{2}-\frac{\Tr[\Pi_1 \rho_1]}{\eta} \bigg| \geq \frac{3}{2} - 2u.
\end{align*}

\subsection{Optimal PPT Measurements for \texorpdfstring{QPV$_{\textsf{SWAP}}$}{QPVswap} Protocol} \label{AppendixSingleRound}

Here we shall prove the upper bound of the success probability of answering the protocol correctly for adversaries restricted to PPT operations in equation \eqref{eqref:LOCC_upperbound}. For simplification we will refer to the equal case as the $0$ case and unequal as the $1$ case. The idea of the proof is to find analytical feasible solutions to the primal and Dual programs of the SDP. In general a feasible solution to the primal program defines a lower bound to the maximization value, while a feasible solution to the dual program defines an upper bound. This is the property of $\textit{weak duality}$ which holds for any SDP \cite{vandenberghe1996semidefinite}. In all of our further proofs we find feasible primal values and dual values that coincide and thus our solutions are optimal and we have $\textit{strong duality}$. 

From the density matrices we see that there is no difference between picking two random equal states or picking two equal states in a random mutually unbiased basis, see $\rho_0$. Similarly picking two random orthogonal states or picking two orthogonal mutually unbiased basis (MUB) states is equal, see $\rho_1$. These become\footnote{Note that this is a slight change of notation with respect to the main text, where we used $\rho_\beta$ for overlap $\beta$. Here $\rho_0$ denotes the mixed state of sending identical states and $\rho_1$ denotes the one sending orthogonal states.}
\begin{align*}
    &\rho_0 = \frac{1}{6} \left( \begin{matrix} 2 & 0 & 0 & 0 \\ 0 & 1 & 1 & 0 \\ 0 & 1 & 1 & 0 \\ 0 & 0 & 0 & 2 \end{matrix} \right),  &\rho_1 = \frac{1}{6} \left( \begin{matrix} 1 & 0 & 0 & 0 \\ 0 & 2 & -1 & 0 \\ 0 & -1 & 2 & 0 \\ 0 & 0 & 0 & 1 \end{matrix} \right).
\end{align*}
It is useful to note that both density matrices $\rho_0, \rho_1$ are a mixture of unentangled states and thereby unentangled. Thus by the Peres-Horodecki separability criterion the partial transpose of $\rho_0$ and $\rho_1$ are positive semi-definite \cite{simon2000peres}. The optimization over all strategies of the single round protocol is written as follows in an SDP:

\begin{minipage}{0.48\textwidth}
\begin{align*}
    &\textbf{Primal Program}\\
    \textbf{maximize: } &\frac{1}{2} \Tr[ \Pi_0 \rho_0 + \Pi_1 \rho_1] \\
    \textbf{subject to: } &\Pi_0 + \Pi_1 = \mathbbm{1}_{2^2} \\
    & \Pi_k \in \text{PPT}(\mathsf{A}:\mathsf{B}), \ \ \ k \in \{0,1\} \\[1ex]
\end{align*}
\end{minipage}
\begin{minipage}{0.48\textwidth}
\begin{align*}
    &\textbf{Dual Program} \\
    \textbf{minimize: } &\Tr[Y] \\
    \textbf{subject to: } &Y - Q^{T_\mathsf{B}}_{i} - \rho_i / 2 \succeq 0, \ \ \ i \in \{0,1\} \\
    & Y\in \text{Herm}(\mathsf{A} \otimes \mathsf{B}) \\
    & Q_i \in \text{Pos}(\mathsf{A}, \mathsf{B}), \ \ \ i \in \{0,1\}. \\
\end{align*}
\end{minipage}
A feasible solution for the primal program is 
\begin{align*}
    &\Pi_0 = \frac{1}{3} \left( \begin{matrix} 2 & 0 & 0 & 0 \\ 0 & 1 & 1 & 0 \\ 0 & 1 & 1 & 0 \\ 0 & 0 & 0 & 2 \end{matrix} \right), &\Pi_1 = \frac{1}{3} \left( \begin{matrix} 1 & 0 & 0 & 0 \\ 0 & 2 & -1 & 0 \\ 0 & -1 & 2 & 0 \\ 0 & 0 & 0 & 1 \end{matrix} \right),
\end{align*}
with solution $ \frac{1}{2} \Tr[ \Pi_0 \rho_0 + \Pi_1 \rho_1] = 2/3$. Note that these measurement projectors correspond to attackers choosing a random MUB to measure in and returning 0 if the measurement outcomes were equal and 1 otherwise, which is also a single round LOCC strategy. This can be seen from the fact that 
\begin{align*}
    \frac{1}{3}( \ket{00}\bra{00} + \ket{11}\bra{11} + \ket{++}\bra{++} + \ket{--}\bra{--} + \ket{i^+i^+}\bra{i^+i^+} + \ket{i^-i^-}\bra{i^-i^-} )&= \Pi_0, \\
    \frac{1}{3}( \ket{10}\bra{10} + \ket{01}\bra{01} + \ket{-+}\bra{-+} + \ket{+-}\bra{+-} + \ket{i^-i^+}\bra{i^-i^+} + \ket{i^+i^-}\bra{i^+i^-} )&= \Pi_1.
\end{align*}
A feasible solution to the dual program is:
\begin{align*}
Y = \frac{\mathbbm{1}_{4}}{6}, && Q_0 = 0 \succeq 0, && Q_1 = \frac{\mathbbm{1}_{4}}{6}  - \frac{\rho_1^{T_\mathsf{B}}}{2} = \frac{1}{12 } \left( \begin{matrix} 1 & 0 & 0 & 1 \\ 0 & 0 & 0 & 0 \\ 0 & 0 & 0 & 0 \\ 1 & 0 & 0 & 1 \end{matrix} \right) = \frac{1}{6} \ket{\Phi^+}\bra{\Phi^+} \succeq 0.
\end{align*}
Which adhere to the constraints in the Dual Program:
\begin{align*}
    Y - Q^{T_\mathsf{B}}_{0} - \frac{\rho_0}{2} &= \frac{\mathbbm{1}_{4}}{6} - \frac{\rho_0}{2} = \frac{1}{12} \left( \begin{matrix} 0 & 0 & 0 & 0 \\ 0 & 1 & -1 & 0 \\ 0 & -1 & 1 & 0 \\ 0 & 0 & 0 & 0 \end{matrix} \right) = \frac{1}{6} \ket{\Psi^-} \bra{\Psi^-} \succeq 0\\
    Y - Q^{T_\mathsf{B}}_{1} - \frac{\rho_1}{2} &= \frac{\mathbbm{1}_{4}}{6} - \left( \frac{\mathbbm{1}_{4}}{6} - \frac{\rho_1}{2} \right) -\frac{\rho_1}{2} = 0 \succeq 0.
\end{align*}
Since $Y\in \text{Herm}(\mathsf{A} \otimes \mathsf{B})$ and we get a feasible solution for the dual with value $\Tr[Y] = \frac{2}{3}$. Thus we have a feasible solution of the primal and dual program that both give the same value so we conclude that the maximal probability of success for attackers under the PPT restriction is $2/3$.

\subsection{Optimal PPT Measurements for \texorpdfstring{QPV$^n_{\textsf{SWAP}}$}{repeated QPVswap} Protocol} \label{AppendixParallel}

We will prove that the optimal probability of succes for attackers in the $n$-round parallel repetition case is $(2/3)^n$. The SDP of the $n$-round parallel repetition protocol is given by:

\begin{minipage}{0.48\textwidth}
\begin{align*}
    &\textbf{Primal Program}\\
    \textbf{maximize: } &\frac{1}{2^n} \sum_{s \in \{0,1\}^n} \Tr[ \Pi_{s} \rho_{s}] \\
    \textbf{subject to: } &\sum_{s \in \{0,1\}^n} \Pi_{s} = \mathbbm{1}_{2^{2n}} \\
    & \Pi_{s} \in \text{PPT}(\mathsf{A}:\mathsf{B}), \ \ \ s \in \{0,1\}^n\\
\end{align*}
\end{minipage}
\begin{minipage}{0.48\textwidth}
\begin{align*}
    &\textbf{Dual Program}\\[1ex]
    \textbf{minimize: } &\Tr[Y] \\[3ex]
    \textbf{subject to: } &Y - Q^{T_\mathsf{B}}_{s} - \rho_{s} / 2^n \succeq 0, \ \ \ s \in \{0,1\}^n \\
    & Y \in \text{Herm}(\mathsf{A} \otimes \mathsf{B}) \\
    & Q_{s} \in \text{Pos}(\mathsf{A} \otimes \mathsf{B}). \\
\end{align*}
\end{minipage}

Repeating the strategy of the single round protocol gives a feasible solution for the primal program with success probability $(2/3)^n$. A feasible solution to the dual problem would yield an upper bound to the problem, but requires finding a general solution for the matrices $Y, Q_{s}$. \\

We start again with by setting $Y$ to be the identity matrix with some proper normalization 
\begin{equation}
    Y = \frac{\mathbbm{1}_{2^{2n}}}{2^{2n}} \left( \frac{2}{3} \right)^n = \frac{\mathbbm{1}_{2^{2n}}}{6^n}, \text{ such that } \Tr[Y] =  \left( \frac{2}{3} \right)^n. 
\end{equation}
We will construct a general feasible solution for $Q_{s}$ for any string $s \in \{0,1\}^n$ from  $Q_{T(s)}$ where $T(s)$ is the reversed sorted version of $s$. First we show a general solution for $s = 0^n$ and $s = 1^n$ string. Again a solution for the all-0 input case is $Q_{0^n} = 0 \succeq 0$. The first constraint for $s = 0^n$ in the dual program of the SDP then reduces to 
\begin{align}\label{allzeroconstraint}
     \frac{\mathbbm{1}_{2^{2n}}}{6^n} - \frac{\rho_0^{\otimes n}}{2^n}.
\end{align}
Note that the eigenvectors of $\rho_0$ are the four Bell states $\{\ket{\Phi^+}, \ket{\Phi^-}, \ket{\Psi^+}, \ket{\Psi^-}\}$, with respective eigenvalues $\{1/3,1/3,1/3,0\}$, then the eigenvalues of $\frac{\rho_0^{\otimes n}}{2^n}$ are $1/6^n$ or $0$. Thus the eigenvalues of \eqref{allzeroconstraint} are either $0$ or $1/6^n$ and \eqref{allzeroconstraint} is positive. Similar to the single round protocol we have the following solution for the $s=1^n$ case
\begin{align*}
    Q_{1^n} &= \frac{\mathbbm{1}_{2^{2n}}}{6^n} - \frac{(\rho_1^{T_\mathsf{B}})^{\otimes n}}{2^n},
\end{align*}
the eigenvectors of $\rho_1^{T_\mathsf{B}}$ are again the Bell states, with respective eigenvalues $\{0,1/3,1/3,1/3\}$. The eigenvectors of $Q_{1^n}$ are all the combinations of tensor products of these four Bell states. If one of these states is the $\ket{\Phi^+}$ state the corresponding eigenvalue of $Q_{1^n}$ is $0$, otherwise the corresponding eigenvalue is $(\frac{1}{6})^n$. Since $Q_{1^n}$ is Hermitian and has only non-negative eigenvalues $Q_{1^n} \succeq 0$, as desired. The corresponding constraint in the dual program of the SDP reduces to
\begin{align*}
    \frac{\mathbbm{1}_{2^{2n}}}{6^n} - \left( \frac{\mathbbm{1}_{2^{2n}}}{6^n} - \frac{\rho_1^{\otimes n}}{2^n} \right ) - \frac{\rho_1^{\otimes n}}{2^n} = 0 \succeq 0.
\end{align*}
We see that the all zero or one case are satisfied. Now suppose we have a valid solution $Q_{s}$ for some $s \in \{0,1\}^n$ and we add a single round of equal inputs, thus
\begin{equation}\label{eq:PrimalStepAssumedPositiveSDP}
    Y - Q^{T_\mathsf{B}}_{s} - \rho_{s} / 2^n \succeq 0.
\end{equation}
And to this $n$-round protocol we add an extra round of equal inputs, we will show that 
\begin{align}\label{eq:iterativesolSDP}
    Q_{s,0} = Q_{s} \otimes \rho_0^{T_\mathsf{B}} / 2,
\end{align} 
is a valid solution for the $(n+1)$-round SDP. We have already shown in Appendix \ref{AppendixSingleRound} that $\rho_0^{T_\mathsf{B}} \succeq 0$. Since the tensor product of positive semi-definite matrices is again positive semi-definite we have that $Q_{s,0} \succeq 0$. Rewriting the first dual constraint we get
\begin{align*}
    \frac{\mathbbm{1}_{2^{2n + 2}}}{6^{n+1}} - Q_{s,0}^{T_\mathsf{B}} - \frac{\rho_{s} \otimes \rho_0}{2^{n+1}} &= \frac{\mathbbm{1}_{2^{2n + 2}}}{6^{n+1}} - Q_{s}^{T_\mathsf{B}} \otimes \frac{\rho_0}{2} - \frac{\rho_{s} \otimes \rho_0}{2^{n+1}} \\
    &= \frac{\mathbbm{1}_{2^{2n}}}{6^{n}} \otimes \frac{\mathbbm{1}_4}{6} - Q_{s}^{T_\mathsf{B}} \otimes \frac{\rho_0}{2} - \frac{\rho_{s} \otimes \rho_0}{2^{n+1}} \\
    &= \frac{\mathbbm{1}_{2^{2n}}}{6^{n}} \otimes \frac{\rho_0 + \rho_1}{3} - Q_{s}^{T_\mathsf{B}} \otimes \frac{\rho_0}{2} - \frac{\rho_{s} \otimes \rho_0}{2^{n+1}} \\
    &= \underbrace{\left( \frac{\mathbbm{1}_{2^{2n}}}{6^{n}} - Q_{s}^{T_\mathsf{B}} - \frac{\rho_{s}}{2^n} \right ) \otimes \frac{\rho_0}{2}}_\textbf{A} + \underbrace{\frac{\mathbbm{1}_{2^{2n}}}{6^{n}} \otimes \left( \frac{2\rho_1 - \rho_0}{6} \right )}_\textbf{B}.
\end{align*}
We see that part $\textbf{A}$ is a tensor product of two positive semi-definite matrices \eqref{eq:PrimalStepAssumedPositiveSDP} and $\rho_0/2$, so $\textbf{A}$ is also positive semi-definite. For part $\textbf{B}$ note that the eigenvectors of $\frac{2\rho_1 - \rho_0}{6}$ are again the Bell states with respective eigenvalues $\{0,0,0,1/6\}$, so part $\textbf{B}$ is positive semi-definite. Since sums of positive semi-definite matrices are positive semi-definite the whole constraint is positive semi-definite. Since for any amount of rounds $n$ we have a feasible solution for the $s=1^n$ case, by repeatedly adding the equal case, we can repeat the previous steps to get a feasible solution for any  reversed sorted string $1^{n}0^k$ for all $n,k$, namely 
\begin{align} \label{eq:ReversedStringPositive}
    Q_{1^n0^k} = Q_{1^n} \otimes \frac{(\rho_0^{T_\mathsf{B}})^{\otimes k}}{2^k}.
\end{align} \\
Now take some string $s \in \{0,1\}^n$, and let $P_s$ be a unitary consisting only of 2-qubit SWAP operations that reverse sorts the $n$-rounds, such that $P_s \rho_{s} P_s^\dagger = \rho_{T(s)}$, and $P^\dagger_s = P_s$. \\

We can now write down the general solution of $Q_s$ using the corresponding map $P_s$ applied to the sorted version. Let $Q_{s} = (P_s Q_{T(s)}^{T_\mathsf{B}} P_s)^{T_\mathsf{B}}$, using the fact that $P$ is a unitary matrix we then get for the corresponding constraint in the dual SDP:
\begin{align*}
    Y - Q^{T_\mathsf{B}}_{s} - \rho_{s} / 2^n \succeq 0 &\Leftrightarrow P_s (Y - Q^{T_\mathsf{B}}_{s} - \rho_{s} / 2^n)P_s \succeq 0 \\
    &\Leftrightarrow Y - P_s Q^{T_\mathsf{B}}_{s} P_s - \rho_{T(s)} / 2^n \succeq 0 \\
    &\Leftrightarrow Y - P_s ((P_s Q^{T_\mathsf{B}}_{T(s)} P_s)^{T_\mathsf{B}})^{T_\mathsf{B}} P_s - \rho_{T(s)} / 2^n \succeq 0\\
    &\Leftrightarrow Y - P_s (P_s Q^{T_\mathsf{B}}_{T(s)} P_s) P_s - \rho_{T(s)} / 2^n \succeq 0 \\
    &\Leftrightarrow Y - Q^{T_\mathsf{B}}_{T(s)} - \rho_{T(s)} / 2^n \succeq 0.
\end{align*}
Where the last expression is positive semi-definite by \eqref{eq:ReversedStringPositive}. Thus we get that the first constraint in the dual program of the $n$-round SDP for any string $s$ is positive semi-definite for any combination of rounds. \\   

The final step is to show that $Q_{s} = (P_s Q_{T(s)}^{T_\mathsf{B}} P_s)^{T_\mathsf{B}}$ is positive. Note that $P_s$ permutes both registers held by $\mathsf{A}$ and $\mathsf{B}$, respectively, of the states together, since it consists only of $2$-qubits SWAP operations. The action is thus independent of the partial transpose on the second party $\mathsf{B}$. We therefore have $Q_{s} = P_s Q_{T(s)} P_s$. Now, since $P_s$ is unitary and $Q_{T(s)}$ is positive semi-definite we have that $Q_s$ is positive semi-definite. \\

We have shown that all the constraints in the dual program of the $n$-round SDP are satisfied by our constructed $Q_{s}$ matrices, thus we have a feasible solution to the dual program with value $\Tr[Y] = (2/3)^n$, which is equal to the primal value and is attainable by a LOCC strategy. This shows that the best attacking strategy for adversaries restricted to LOCC operations playing $n$ rounds in parallel is to simply apply the single round strategy $n$ times in parallel.

\subsection{Optimal PPT Measurements for loss-tolerant \texorpdfstring{$\text{QPV}^n_{\textsf{SWAP}}$}{qpvswapnlossy} Protocol} \label{AppendixLossTolerantParallelRepetition}

We shall now modify the solution to the parallel repetition case in Appendix \ref{AppendixParallel} to give a solution to the maximization of conditional success probability under LOCC restrictions. We will optimize the probability of being correct conditioned on answering. The SDP for the lossy $n$ round parallel repetition protocol in which attackers either answer on all rounds or on none is given as:

\begin{minipage}{0.45\textwidth}
\begin{align*}
    &\textbf{Primal Program}\\
    \textbf{maximize: } &\frac{1}{2^n \eta} \sum_{s \in \{0,1\}^n} \Tr[ \tilde{\Pi}_{s} \rho_{s}] \\
    \textbf{subject to: } & \left(\sum_{s \in \{0,1\}^n} \tilde{\Pi}_{s} \right) + \tilde{\Pi}_\varnothing = \mathbbm{1}_{2^{2n}} \\
    & \Tr[\tilde{\Pi}_\varnothing \rho_s] = 1 - \eta, \ \ \ s \in \{0,1\}^n \\
    & \tilde{\Pi}_{s} \in \text{PPT}(\mathsf{A}:\mathsf{B}), \ \ \ s \in \{0,1\}^n \cup \varnothing \\
\end{align*}
\end{minipage}
\begin{minipage}{0.45\textwidth}
\begin{align*}
    &\textbf{Dual Program}\\
    \textbf{minimize: } & \frac{\Tr[\tilde{Y}] - (1 - \eta) \gamma}{\eta} \\
    \textbf{subject to: } &\tilde{Y} - \tilde{Q}^{T_\mathsf{B}}_{s} - \rho_{s} / 2^n \succeq 0, \ \ \ s \in \{0,1\}^n \\
    & 2^{2n} ( \tilde{Y} - \tilde{Q}^{T_\mathsf{B}}_\varnothing ) - \gamma \mathbbm{1}_{2^{2n}} \succeq 0 \\
    & \tilde{Y} \in \text{Herm}(\mathsf{A} \otimes \mathsf{B}) \\
    & \tilde{Q}_{s} \in \text{Pos}(\mathsf{A} \otimes \mathsf{B}), \ \ \ s \in \{0,1\}^n \cup \varnothing \\
    & \gamma \in \mathbb{R}. \\
\end{align*}
\end{minipage}

\noindent Here $\eta$ is the transmission rate and $\Tr[\tilde{\Pi}_\varnothing \rho_s] = 1 - \eta$ is the condition that attackers can only say loss with equal probability on every input. We suspect our protocol is loss-tolerant, thus we want the solution to be independent of $\eta$. It turns out multiplying the POVM elements by $\eta$ and picking $\tilde{\Pi}_\varnothing$ accordingly, i.e.\ $\tilde{\Pi}_s = \eta \Pi_s$ for every $s \in \{0,1\}^n$ and $\tilde{\Pi}_\varnothing = (1 - \eta) \mathbbm{1}_{2^{2n}}$ gives a feasible solution for the primal program with solution $(2/3)^n$.

For the dual program, we pick
\begin{align}
    \tilde{Y} = \frac{\mathbbm{1}_{2^{2n}}}{6^n}, &&\tilde{Q_s} = Q_s &&\tilde{Q}_\varnothing = 0, &&\gamma = (2/3)^n,
\end{align}
then trivially $Y \in \text{Herm}(\mathsf{A} \otimes \mathsf{B}), \tilde{Q}_{s} \in \text{Pos}(\mathsf{A} \otimes \mathsf{B}), \gamma \in \mathbb{R}$ and the first condition remains satisfied since we have not changed $Y, Q_s$ in Appendix \ref{AppendixParallel}. The second constraint becomes 
\begin{align}
    2^{2n} ( \tilde{Y} - \tilde{Q}^{T_\mathsf{B}}_\varnothing ) - \gamma \mathbbm{1}_{2^{2n}} &= \mathbbm{1}_{2^{2n}}\frac{2^n}{3^n} - (2/3)^n \mathbbm{1}_{2^{2n}} = 0 \succeq 0.
\end{align}
So all constraints in the dual are satisfied. We thus get an upper bound of
\begin{align}
    \frac{\Tr[\tilde{Y}] - (1 - \eta) \gamma}{\eta} = \frac{(2/3)^n - (1-\eta) (2/3)^n}{\eta} = \frac{\eta (2/3)^n}{\eta} = (2/3)^n.
\end{align}
Thus we finally have $p_{\text{succ},n}^{\text{max}}(\eta) = (2/3)^n$ for any $\eta \in (0, 1]$. Together with Proposition \ref{prop:lossy_swap_allornothing} in the main text this gives full loss tolerance for the $n$-round parallel repetition of our protocol.

\section{Explicit descriptions for \texorpdfstring{$\Prob_{\Omega_\beta}(\mathsf{D}_1, \mathsf{D}_2)$}{Pd1d2} and \texorpdfstring{$\Prob_{\Omega_\beta}(\mathsf{D}_1, \mathsf{D}_4)$}{Pd1d4}}
\label{app:Pd1d2_Pd1d4}

As argued in the main text, we only need to find expressions for $\Prob_{\Omega_\beta}(\mathsf{D}_1, \mathsf{D}_2)$ and $\Prob_{\Omega_\beta}(\mathsf{D}_1, \mathsf{D}_4)$ in terms of experimental parameters. Then we can recover $\Prob_{\Omega_\beta}(0)$ as well as $\Prob_{\Omega_\beta}(1)$ and therefore also $\Prob_{\Omega_\beta}(0 \,|\, \text{concl.})$ and $\Prob_{\Omega_\beta}(1 \,|\, \text{concl.})$, which are the probabilities of interest for security analysis. In what follows we will find explicit expressions of each term in our expansion

\begin{equation}
	\begin{alignedat}{3}\label{appequ: Pd1d2}
		&\Prob_{\Omega_\beta}(\mathsf{D}_1, \mathsf{D}_2) &&= \,\,&&(1-p_\text{dark})^2 \sum_{k} \Prob_{\Omega_\beta}(\mathsf{D}_1, \mathsf{D}_2 \,|\, k \text{ photons at } \mathsf{BS}_1) \Prob_{\Omega_\beta}(k \text{ photons at } \mathsf{BS}_1) \\
		&{}&&= \,\,&&(1-p_\text{dark})^2 \sum_k \bigg[ \Prob_{\Omega_\beta}(\mathsf{D}_1, \mathsf{D}_2 \,|\,\text{bunch}, k) \frac{\Prob_{\Omega_\beta}(\text{bunch} \,|\, k )}{2} \\
		&{}&&{}&&+ \Prob_{\Omega_\beta}(\mathsf{D}_1, \mathsf{D}_2 \,|\,\text{anti-bunch}, k)\Prob_{\Omega_\beta}(\text{anti-bunch} \,|\, k) \bigg] \Prob_{\Omega_\beta}(k \text{ photons at } \mathsf{BS}_1),
	\end{alignedat}
\end{equation}
and the analogous formula for $\Prob_{\Omega_\beta}(\mathsf{D}_1, \mathsf{D}_4)$. We will first treat the terms that are part of both probabilities $\Prob_{\Omega_\beta}(\mathsf{D}_1, \mathsf{D}_2)$ and $\Prob_{\Omega_\beta}(\mathsf{D}_1, \mathsf{D}_4)$. Note that the sources produce $\ell \leq 3$ photons with the following probabilities $p_\ell$:
	\begin{align*}
		p_0 &= (1- \eta_\text{source})^2, \\
		p_1 &= 2B(1-\eta_\text{source}), \\
		p_2 &= B^2 + 2p_\text{pair}(1-\eta_\text{source}), \\
		p_3 &= 2 p_\text{pair}B, \\
		p_4 &= p_\text{pair}^2
	\end{align*}
Then the probabilities that $k$ photons arrive at $\mathsf{BS}_1$ given $\ell \leq 3$ photons produced is
	\begin{align}
		\Prob_{\Omega_\beta}(k \text{ photons at } \mathsf{BS}_1) &= \sum_{\ell = k}^{3} \Prob_{\Omega_\beta}(k \text{ photons at } \mathsf{BS}_1 \,|\, \ell \text{ produced}) \Prob_{\Omega_\beta}(\ell \text{ produced}) \\
		&= \sum_{\ell = k}^{3} \binom{\ell}{k} \eta_\text{BS}^k (1-\eta_\text{BS})^{\ell-k} p_\ell.
	\end{align}
Next, we consider the probabilities to bunch or anti-bunch given $k$ photons interfering at the beam splitter. To that end, we first had to derive the output port distribution for 3 incoming photons (2 from one side, 1 from the other)\footnote{This will be a generalisation of the Hong-Ou-Mandel output port distribution $\Prob((0,2) \text{ or } (2,0)) = \frac{1 + \lvert \braket{\psi|\phi}\rvert^2}{2}$ to 3 photons.}, see lemma \ref{HOM3} in the main text and appendix \ref{app:3HOM_proof} for the proof. From lemma \ref{HOM3} we gather that for $\ket{\psi} = \ket{\phi}$ being identical and overlap $\beta = \lvert \braket{\psi|\chi} \rvert$ we have
\begin{align}
	\Prob_\text{ideal}(\text{bunch} \,|\, 3) &= 4 \lvert R \rvert^2 \lvert T \rvert^2 \frac{1 + 2 \beta^2}{4}, \\
	\Prob_\text{ideal}(\text{anti-bunch} \,|\, 3) &= 1 - \Prob_\text{ideal}(\text{bunch} \,|\, 3).
\end{align}
In the imperfect case we have to consider all cases that can appear with $k$ incoming photons, such as 2 photons in one input port, or 1 in one and 2 in the other input port. First, the case of $k=0$ does not matter because we neglect terms proportional to $p_\text{dark}^2$. The case of $k=1$ is trivial as one could say the photon always ``bunches''. Hence we can set
\begin{align}
	\Prob_{\Omega_{\beta}}(\text{bunch} \,|\, 1) &= 1/\Prob_{\Omega_{\beta}}(1 \text{ photon at } \mathsf{BS}_1).
\end{align}
For two photons we need to distinguish between the cases of both photons coming into the same input port (no interference) or one photon in each input port (interference). The respective probabilities are
\begin{align}
	\Prob_{\Omega_{\beta}}(\text{all in one mode} \,|\, 2) &= \frac{2p_\text{pair}(1-\eta_\text{source})\eta_\text{BS}^2+2p_\text{pair}B\eta_\text{BS}^2(1-\eta_\text{BS})+p_\text{pair}^2\eta_\text{BS}^2(1-\eta_\text{BS})^2}{\Prob_{\Omega_{\beta}}(2 \text{ photons at } \mathsf{BS}_1)}, \\
	\Prob_{\Omega_{\beta}}(\text{one in each mode} \,|\, 2) &= 1 - \Prob(\text{all in one mode} \,|\, 2).
\end{align}
Then the overall probability to bunch given 2 incoming photons is
\begin{align}
	\Prob_{\Omega_{\beta}}(\text{bunch} \,|\, 2) &= \left(\lvert R \rvert^4 + \lvert T \rvert^4\right)\Prob_{\Omega_{\beta}}(\text{all in one mode} \,|\, 2) + 4 \lvert R \rvert^2 \lvert T \rvert^2 \frac{1 + \beta^2}{2} \Prob_{\Omega_{\beta}}(\text{one in each mode} \,|\, 2).
\end{align}
For 3 photons we get
\begin{align}
	\Prob_{\Omega_{\beta}}(\text{bunch} \,|\, 3) = 4 \lvert R \rvert^2 \lvert T \rvert^2 \frac{1 + 2 \beta^2}{4}.
\end{align}
Finally, note that a single photon click/no-click detector clicks if at least one photon triggers it \cite{migdall2013single}. Hence the probability for a click gets higher if more than one photon reach the detector. To account for that we define $p_\text{click}(m) = \Prob(\text{photon 1 detected } \cup \dots \cup \text{ photon } m \text{ detected})$, describing the probability that a detector clicks if $m$ photons go into it. This can be expanded via the inclusion-exclusion principle and the independence of the events of each photon being detected. Fundamentally we parametrize $\Prob(\text{photon } x \text{ detected}) = \eta_\text{det}$ so that $p_\text{click}(m)$ is a function of $\eta_\text{det}$ only. For completeness we give them here:
\begin{align}
	p_\text{click}(1) &= \eta_\text{det}, \\
	p_\text{click}(2) &= 2\eta_\text{det}-\eta_\text{det}^2,\\
	p_\text{click}(3) &= 3\eta_\text{det}-3\eta_\text{det}^2+\eta_\text{det}^3.
\end{align}
We now continue with the separate expressions for detectors $(\mathsf{D_1}, \mathsf{D_2})$ and $(\mathsf{D_1}, \mathsf{D_4})$ respectively.

\subsection{Clicking probabilities for \texorpdfstring{$(\mathsf{D_1}, \mathsf{D_2})$}{Probsd1d2}}
Again, we will distinguish the cases for different numbers of $k\leq3$ interfering photons. Since we condition on $k$ photons having bunched or anti-bunched, we just need to go through the next beam splitter (where no interference happens) and add up the events for which we get a $(\mathsf{D_1}, \mathsf{D_2})$ click pattern. In the bunching case we may assume that the photons bunched into the $(\mathsf{D_1}, \mathsf{D_2})$ arm, because otherwise two dark counts would be needed, dominated by a factor $p_\text{dark}^2$. The results are
\begin{align}
	\Prob_{\Omega_{\beta}}(\mathsf{D_1}, \mathsf{D_2} \,|\, \text{bunch}, 0) &= O\left(p_\text{dark}^2\right) \\
	\Prob_{\Omega_{\beta}}(\mathsf{D_1}, \mathsf{D_2} \,|\, \text{bunch}, 1) &= p_\text{click}(1)p_\text{dark} + O\left(p_\text{dark}^2\right), \\
	\Prob_{\Omega_{\beta}}(\mathsf{D_1}, \mathsf{D_2} \,|\, \text{bunch}, 2) &= 2\lvert R \rvert^2 \lvert T \rvert^2 p_\text{click}(1)^2 + \left( \lvert R \rvert^4 + \lvert T \rvert^4 \right) p_\text{click}(2) p_\text{dark} + O\left(p_\text{dark}^2\right), \\
	\Prob_{\Omega_{\beta}}(\mathsf{D_1}, \mathsf{D_2} \,|\, \text{bunch}, 3) &= 3 \left( \lvert R \rvert^4 \lvert T \rvert^2 + \lvert R \rvert^2 \lvert T \rvert^4 \right) p_\text{click}(1) p_\text{click}(2) + \left( \lvert R \rvert^6 + \lvert T \rvert^6 \right) p_\text{click}(3) p_\text{dark} \\
	& \quad + O\left(p_\text{dark}^2\right).
\end{align}
Similarly, one gets\footnote{For $\Prob_{\Omega_{\beta}}(\mathsf{D_1}, \mathsf{D_2} \,|\, \text{anti-bunch}, 1)$ we consider the photon to be leaving into the $(\mathsf{D_3}, \mathsf{D_4})$ arm because we already have the $(\mathsf{D_1}, \mathsf{D_2})$ case in $\Prob_{\Omega_{\beta}}(\mathsf{D_1}, \mathsf{D_2} \,|\, \text{bunch}, 1)$. This again incurs a factor $p_\text{dark}^2$.}
\begin{align}
	\Prob_{\Omega_{\beta}}(\mathsf{D_1}, \mathsf{D_2} \,|\, \text{anti-bunch}, 0) &= O\left(p_\text{dark}^2\right) \\
	\Prob_{\Omega_{\beta}}(\mathsf{D_1}, \mathsf{D_2} \,|\, \text{anti-bunch}, 1) &= O\left(p_\text{dark}^2\right), \\
	\Prob_{\Omega_{\beta}}(\mathsf{D_1}, \mathsf{D_2} \, | \, \text{anti-bunch}, 2) &= (1-p_\text{click}(1))p_\text{click}(1)p_\text{dark} + O\left(p_\text{dark}^2\right), \\
	\Prob_{\Omega_{\beta}}(\mathsf{D_1}, \mathsf{D_2} \, | \, \text{anti-bunch}, 3) &= \frac{1}{2} \left( 2 \lvert R \rvert^2 \lvert T \rvert^2 \left( 1- p_\text{click}(1) \right)^2 + \left( \lvert R \rvert^4 + \lvert T \rvert^4 \right) \left( 1- p_\text{click}(2) \right) \right) \\
	& \quad \cdot p_\text{click}(1)p_\text{dark} + \frac{1}{2} \left( 1 - p_\text{click}(1) \right) \Prob_{\Omega_{\beta}}(\mathsf{D_1}, \mathsf{D_2} \,|\, \text{bunch}, 2) + O\left(p_\text{dark}^2\right).
\end{align}

\subsection{Clicking probabilities for \texorpdfstring{$(\mathsf{D_1}, \mathsf{D_4})$}{Probsd1d4}}
Analogously we treat the case for the $(\mathsf{D_1}, \mathsf{D_4})$ click pattern. This yields
\begin{align}
	\Prob_{\Omega_{\beta}}(\mathsf{D_1}, \mathsf{D_4} \,|\, \text{bunch}, 1) &= \lvert T \rvert^2 p_\text{click}(1)p_\text{dark} + O\left(p_\text{dark}^2\right), \\
	\Prob_{\Omega_{\beta}}(\mathsf{D_1}, \mathsf{D_4} \,|\, \text{bunch}, 2) &= \lvert T \rvert^4 p_\text{click}(2) p_\text{dark} + 2 \lvert R \rvert^2 \lvert T \rvert^2 p_\text{click}(1) \left(1 - p_\text{click}(1) \right) p_\text{dark} + O\left(p_\text{dark}^2\right), \\
	\Prob_{\Omega_{\beta}}(\mathsf{D_1}, \mathsf{D_4} \,|\, \text{bunch}, 3) &= \lvert T \rvert^6 p_\text{click}(3) p_\text{dark} + 3 \lvert R \rvert^2 \lvert T \rvert^4 p_\text{click}(2) \left( 1 - p_\text{click}(1)\right)p_\text{dark} \\
	&\quad + 3 \lvert R \rvert^4 \lvert T \rvert^2 p_\text{click}(1) \left( 1 - p_\text{click}(2)\right)p_\text{dark} + O\left(p_\text{dark}^2\right).
\end{align}
And for the anti-bunching cases,
\begin{align}
	\Prob_{\Omega_{\beta}}(\mathsf{D_1}, \mathsf{D_4} \,|\, \text{anti-bunch}, 0) &= O\left(p_\text{dark}^2\right) \\ \Prob_{\Omega_{\beta}}(\mathsf{D_1}, \mathsf{D_4} \,|\, \text{anti-bunch}, 1) &= O\left(p_\text{dark}^2\right), \\
	\Prob_{\Omega_{\beta}}(\mathsf{D_1}, \mathsf{D_4} \, | \, \text{anti-bunch}, 2) &= \lvert T \rvert^4 p_\text{click}(1)^2 + 2 \lvert R \rvert^2 \lvert T \rvert^2 p_\text{click}(1) \left(1 - p_\text{click}(1)\right) p_\text{dark} + O\left(p_\text{dark}^2\right), \\
	\Prob_{\Omega_{\beta}}(\mathsf{D_1}, \mathsf{D_4} \, | \, \text{anti-bunch}, 3) &= \lvert T \rvert^6 p_\text{click}(2) + 2 \lvert R \rvert^2 \lvert T \rvert^4 p_\text{click}(1)^2 \left( 1 - p_\text{click}(1) \right) \\
	&\quad+ \left( \lvert R \rvert^4 \lvert T \rvert^2 + \lvert T \rvert^6 \right) p_\text{click}(1) \left( 1 - p_\text{click}(2) \right) p_\text{dark} \\
	&\quad+ \left( 2 \lvert R \rvert^2 \lvert T \rvert^4 + 2 \lvert R \rvert^4 \lvert T \rvert^2 \right) p_\text{click}(1) \left(1 - p_\text{click}(1)\right)^2 p_\text{dark} \\
	&\quad+\lvert R \rvert^2 \lvert T \rvert^4 p_\text{click}(2) \left( 1 - p_\text{click}(1) \right) p_\text{dark} + O\left(p_\text{dark}^2\right).
\end{align}
Now we have expanded all parts of equations \eqref{equ: Pd1d2} and \eqref{equ: Pd1d4}.

\section{Original SWAP-test setup with one beam splitter}
\label{app:Pd1_Pd1d2_swap}
We have done the analogous expansions as in the previous appendix \ref{app:Pd1d2_Pd1d4} also for the original setup of the SWAP-test with just one $(R,T)$ beam splitter and two detectors $\mathsf{D}_1$, $\mathsf{D}_2$. For brevity we won't include all the formulae here, but they look very similar to the ones in \ref{app:Pd1d2_Pd1d4}. In this case we break the problem down to finding $\Prob_{\Omega_{\beta}}(\mathsf{D_1})$ and $\Prob_{\Omega_{\beta}}(\mathsf{D_1}, \mathsf{D_2})$, which are similarly expanded as in \eqref{equ: Pd1d2} and \eqref{equ: Pd1d4}. 

\section{Proof of 3-photon output port distribution}
\label{app:3HOM_proof}
We generalise the well known output probability distribution of the HOM effect after 2 photons entered a (symmetric) $(R,T)$ beam splitter in the two different input ports and are detected in the output ports. We denote detector clicks as $(c_1, c_2)$ with $c_k$ indicating the number of photons registered at detector $k$. Then, if a photonic qu$d$it in state $\ket{\psi}$ enters the beam splitter from one input port and $\ket{\phi}$ does so from the other, one gets \cite{loudon2000quantum}
\begin{align}
	\Prob((2,0) \text{ or } (0,2)) = 4 \lvert R \rvert^2 \lvert T \rvert^2 \frac{1+\lvert \braket{\psi|\phi}\rvert^2}{2} \\
	\Prob((1,1)) = 1 - \Prob((2,0) \text{ or } (0,2)).
\end{align}

\noindent Here we generalise this to the 3-photon case, yielding the following lemma.

\begin{lemma}
	Consider photonic qubits $\ket{\psi}, \ket{\phi}$ arriving at one input port of a (symmetric) $(R,T)$ beam splitter and $\ket{\chi}$ at the other input port. Then the output port distribution is given by
	\begin{align}\label{HOM3_appendix}
		p_\emph{bunch} = \Prob((3,0) \emph{ or } (0,3)) &= 4 \lvert R \rvert^2 \lvert T \rvert^2 \frac{\lvert \braket{\psi|\phi} \rvert^2 + \lvert\braket{\psi|\chi}\rvert^2 + \lvert\braket{\phi|\chi}\rvert^2}{2 \cdot (1 + \lvert\braket{\psi|\phi}\rvert^2)}, \\
		p_\emph{anti-bunch} = \Prob((2,1) \emph{ or } (1,2)) &= 1-\Prob((3,0) \emph{ or } (0,3)).
	\end{align} 
\end{lemma}

\begin{proof}
	For notational simplicity we give the proof for the 50/50 beam splitter case, for which $\lvert R \rvert^2 = \lvert T \rvert^2 = 1/2$. The same calculation can be done with general coefficients $(R,T)$. Let there be 3 incoming photonic qubits in the states
	\begin{align}
		\ket{\psi} &= \alpha_0 \ket{0} + \alpha_1 \ket{1} \\
		\ket{\phi} &= \beta_0 \ket{0} + \beta_1 \ket{1} \\
		\ket{\chi} &= \gamma_0 \ket{0} + \gamma_1 \ket{1},
	\end{align}
	with $\ket{\psi}$ and $\ket{\phi}$ entering one input port, and $\ket{\chi}$ entering the other. For simplicity, we consider photons with $H/V$ polarization as the qubit basis states. The spatial modes going into the input ports of the beam splitter are denoted by $a, b$. We first need to write down the (normalised by $\mathcal{N}$) input Fock state, which is
	\begin{align*}
		\ket{\text{in}} &= &&\mathcal{N} (\alpha_0  \mathtt{a_H}^\dag + \alpha_1 \mathtt{a_V}^\dag)(\beta_0 \mathtt{a_H}^\dag + \beta_1 \mathtt{a_V}^\dag)(\gamma_0 \mathtt{b_H}^\dag + \gamma_1 \mathtt{b_V}^\dag)\ket{0_a, 0_b} \\
		&= &&\mathcal{N} \Big[ \sqrt{2} \alpha_0 \beta_0 \gamma_0 \ket{2_H, 1_H} +  \sqrt{2} \alpha_0 \beta_0 \gamma_1 \ket{2_H, 1_V} + ( \alpha_0 \beta_1 + \alpha_1 \beta_0 ) \gamma_0 \ket{1_H 1_V, 1_H} + \\
		&{}&&( \alpha_0 \beta_1 + \alpha_1 \beta_0 ) \gamma_1 \ket{1_H 1_V, 1_V} + \sqrt{2} \alpha_1 \beta_1 \gamma_0 \ket{2_V, 1_H} +  \sqrt{2} \alpha_1 \beta_1 \gamma_1 \ket{2_V, 1_V} \Big].
	\end{align*}
	Requiring $\lVert \ket{\text{in}} \rVert = 1$ gives an expression for the normalisation constant in terms of amplitudes that can be rewritten as $\mathcal{N} = \frac{1}{\sqrt{1+\lvert \braket{\psi|\phi} \rvert^2}}$. We needed $\ket{\text{in}}$ in terms of actual Fock states, as in the second line above, in order to be able to find the normalisation factor $\mathcal{N}$. When passing through it, the $50/50$ beam splitter acts as a unitary $U$ on the photonic creation operators as
	\begin{align}
		\mathtt{a_H}^\dag &\mapsto \frac{\mathtt{c_H}^\dag + i \mathtt{d_H}^\dag }{\sqrt{2}} &&\mathtt{b_H}^\dag \mapsto \frac{i\mathtt{c_H}^\dag + \mathtt{d_H}^\dag }{\sqrt{2}} \\
		\mathtt{a_V}^\dag &\mapsto \frac{\mathtt{c_V}^\dag + i \mathtt{d_V}^\dag }{\sqrt{2}} &&\mathtt{b_V}^\dag \mapsto \frac{i\mathtt{c_V}^\dag + \mathtt{d_V}^\dag }{\sqrt{2}}.
	\end{align}
	Therefore, after a considerable amount of algebra, we arrive at
	\begin{align*}
		\ket{\text{in}} \mapsto \ket{\text{out}} = \frac{\mathcal{N}}{2 \sqrt{2}} \bigg[ & \textcolor{QuSoft}{i \sqrt{6} \alpha_0 \beta_0 \gamma_0 \ket{3_H, 0} + i \sqrt{2} (\alpha_1 \beta_0 \gamma_0 + \alpha_0 \beta_1 \gamma_0 + \alpha_0 \beta_0 \gamma_1) \ket{2_H 1_V, 0} +} \\
		& \textcolor{QuSoft}{i \sqrt{2} (\alpha_1 \beta_1 \gamma_0 + \alpha_1 \beta_0 \gamma_1 + \alpha_0 \beta_1 \gamma_1) \ket{1_H 2_V, 0} + i \sqrt{6} \alpha_1 \beta_1 \gamma_1 \ket{3_V, 0}} \textcolor{UltramarineBlue}{-} \\
		& \textcolor{UltramarineBlue}{\sqrt{2} \alpha_0 \beta_0 \gamma_0 \ket{2_H, 1_H} - \sqrt{2} \alpha_0 \beta_0 \gamma_1 \ket{1_H 1_V, 1_H} +} \\
		& \textcolor{UltramarineBlue}{\sqrt{2}(\alpha_1 \beta_1 \gamma_0 - \alpha_1 \beta_0 \gamma_1 - \alpha_0 \beta_1 \gamma_1) \ket{2_V, 1_H} + i \sqrt{2} \alpha_0 \beta_0 \gamma_0 \ket{1_H, 2_H} +} \\
		& \textcolor{UltramarineBlue}{i \sqrt{2} (\alpha_1 \beta_0 \gamma_0 + \alpha_0 \beta_1 \gamma_0 - \alpha_0 \beta_0 \gamma_1) \ket{1_V, 2_H}} \textcolor{QuSoft}{- \sqrt{6} \alpha_0 \beta_0 \gamma_0 \ket{0, 3_H}} \textcolor{UltramarineBlue}{+} \\
		& \textcolor{UltramarineBlue}{\sqrt{2}(\alpha_0 \beta_0 \gamma_1 - \alpha_1 \beta_0 \gamma_0 - \alpha_0 \beta_1 \gamma_0) \ket{2_H, 1_V} - 2i \alpha_1 \beta_1 \gamma_0 \ket{1_H 1_V, 1_V} -} \\
		& \textcolor{UltramarineBlue}{\sqrt{2} \alpha_1 \beta_1 \gamma_1 \ket{2_V, 1_V} + 2i \alpha_0 \beta_0 \gamma_1 \ket{1_H, 1_H 1_V} + 2i \alpha_1 \beta_1 \gamma_0 \ket{1_V, 1_H 1_V}} \textcolor{QuSoft}{-} \\
		& \textcolor{QuSoft}{\sqrt{2}(\alpha_1 \beta_0 \gamma_0 + \alpha_0 \beta_1 \gamma_0 + \alpha_0 \beta_0 \gamma_1) \ket{0, 2_H 1_V}} \textcolor{UltramarineBlue}{+} \\
		& \textcolor{UltramarineBlue}{i \sqrt{2} (\alpha_1 \beta_0 \gamma_1 + \alpha_0 \beta_1 \gamma_1 - \alpha_1 \beta_1 \gamma_0) \ket{1_H, 2_V} + i \sqrt{2} \alpha_1 \beta_1 \gamma_1 \ket{1_V, 2_V}} \textcolor{QuSoft}{-} \\
		& \textcolor{QuSoft}{\sqrt{2}(\alpha_1 \beta_1 \gamma_0 + \alpha_1 \beta_0 \gamma_1 + \alpha_0 \beta_1 \gamma_1) \ket{0, 1_H 2_V} - \sqrt{6} \alpha_1 \beta_1 \gamma_1 \ket{0, 3_V}} \bigg],
	\end{align*}
	where the \textcolor{QuSoft}{red terms} indicate states with all 3 photons in one detector arm, and the \textcolor{UltramarineBlue}{blue terms} states with photons in both detector arms. One can check that $\ket{\text{out}}$ is normalised and thus indeed a valid quantum state. This yields
	\begin{align*}
		\Prob((3,0) \text{ or } (0,3)) = \frac{\mathcal{N}^2}{2} \Big[ &3 |\alpha_0|^2 |\beta_0|^2 |\gamma_0|^2 + 3 |\alpha_1|^2 |\beta_1|^2 |\gamma_1|^2 + \\
		& | (\alpha_0 \beta_1 + \alpha_1 \beta_0)\gamma_0 + \alpha_0 \beta_0 \gamma_1 |^2 + | (\alpha_0 \beta_1 + \alpha_1 \beta_0)\gamma_1 + \alpha_1 \beta_1 \gamma_0 |^2 \Big].
	\end{align*}
	Writing the coefficients $\alpha_k$, $\beta_k$, $\gamma_k$ in terms of their respective Bloch angles $(\varphi, \theta)$ simplifies this, after some algebra, to
	\begin{align*}
		\Prob((3,0) \text{ or } (0,3)) = \mathcal{N}^2 \bigg[ \frac{3}{4} + \frac{1}{4} ( \mathbf{r}_\psi \cdot \mathbf{r}_\phi + \mathbf{r}_\psi \cdot \mathbf{r}_\chi + \mathbf{r}_\phi \cdot \mathbf{r}_\chi) \bigg],
	\end{align*}
	with the Bloch vectors $\mathbf{r}$ corresponding to their respective states. We now use the correspondence between the dot product between Bloch vectors and the inner product between Hilbert space states, namely
	\begin{align}
		\lvert \braket{\alpha|\beta} \rvert^2 = \frac{1}{2} ( 1 + \mathbf{r}_\alpha \cdot \mathbf{r}_\beta).
	\end{align}
	Inserting everything into the above probability and simplifying finally yields
	\begin{align}
		\Prob((3,0) \text{ or } (0,3)) = \frac{\lvert \braket{\psi|\phi} \rvert^2 + \lvert\braket{\psi|\chi}\rvert^2 + \lvert\braket{\phi|\chi}\rvert^2}{2 \cdot (1 + \lvert\braket{\psi|\phi}\rvert^2)},
	\end{align}
	where we also inserted $\mathcal{N}$. Accordingly, this also gives us
	\begin{align}
		\Prob((2,1) \text{ or } (1,2)) &= 1 - \Prob((3,0) \text{ or } (0,3))
	\end{align}
\end{proof}

\end{appendices}
\end{document}